\documentclass[11pt]{article}
\usepackage[margin=25mm]{geometry}

\usepackage{mathtools, amssymb}
\usepackage{amsthm}
\usepackage[scr=rsfs]{mathalpha}
\usepackage{nicefrac}

\usepackage{titlesec}
\usepackage{enumitem}
\usepackage{caption}
\usepackage[symbol]{footmisc}
\usepackage{abstract}
\usepackage{authblk}

\usepackage{tikz}

\usepackage[
    backend=biber, 
    style=authoryear-comp,
    uniquename=init,
    giveninits=true,
    date=year,
]{biblatex}
\AtEveryBibitem{%
    \clearlist{location}%
    \clearfield{note}%
    \clearfield{urldate}%
    \clearfield{urlyear}%
    \ifentrytype{software}{}{%
        \clearfield{url}%
    }
}
\DeclareFieldFormat{pages}{#1}
\renewbibmacro{in:}{}
\usepackage{hyperref}
\hypersetup{
  colorlinks=true, 
  urlcolor=blue, 
  linkcolor=blue, 
  citecolor=blue 
}
\usepackage{cleveref}
\crefname{equation}{equation}{equations}

\title{The determining role of covariances in large networks of stochastic neurons}

\author[$*$]{Vincent Painchaud}
\author[$\dagger$,$\mathsection$,$\mathparagraph$]{Patrick Desrosiers}
\author[$\ddagger$,$\mathsection$,$\mathparagraph$]{Nicolas Doyon}

\affil[$*$]{Department of Mathematics and Statistics, McGill University, Montreal, Québec, Canada H3A 0B6}
\affil[$\dagger$]{Départment de physique, de génie physique et d’optique, Université Laval, Quebec City, Québec, Canada G1V 0A6}
\affil[$\ddagger$]{Départment de mathématiques et de statistique, Université Laval, Quebec City, Québec, Canada G1V 0A6}
\affil[$\mathsection$]{CERVO Brain Research Center, Quebec City, Québec, Canada G1E 1T2}
\affil[$\mathparagraph$]{Centre interdisciplinaire en modélisation mathématique de l’Université Laval, Quebec City, Québec, Canada G1V 0A6}

\date{}

\titleformat{\section}{\large\bfseries}{\thesection}{1em}{}
\titleformat{\subsection}{\bfseries}{\thesubsection}{1em}{}
\titleformat{\subsubsection}{\bfseries}{\thesubsubsection}{1em}{}
\newcommand*{\enumlabelformat}[1]{%
    \textnormal{\textit{#1})}%
}
\setlist[enumerate]{
    nosep,
    label=\enumlabelformat{\roman*},
}

\numberwithin{equation}{section}

\newtheoremstyle{plain}
  {\topsep}   % ABOVESPACE
  {\topsep}   % BELOWSPACE
  {\itshape}  % BODYFONT
  {0pt}       % INDENT
  {\bfseries} % HEADFONT
  {.}         % HEADPUNCT
  {5pt plus 1pt minus 1pt} % HEADSPACE
  {}          % CUSTOM-HEAD-SPEC
\newtheorem{theorem}{Theorem}
\newtheorem{corollary}{Corollary}
\newtheorem{proposition}{Proposition}

\newcommand*{\charf}[1]{\mathbf{1}_{#1}}

\newcommand*{\beq}[1]{\begin{equation}\label{#1}}
\newcommand*{\eeq}{\end{equation}}

\renewcommand*{\Re}{\operatorname{Re}}
\renewcommand*{\Im}{\operatorname{Im}}
\newcommand*{\defeq}{\vcentcolon=}
\newcommand*{\eqdef}{=\vcentcolon}

\DeclarePairedDelimiter{\abs}{\lvert}{\rvert}
\DeclarePairedDelimiterX{\inprod}[2]{\langle}{\rangle}{#1,#2}
\DeclarePairedDelimiterXPP{\bprob}[1]{\mathbb{P}^\eta}{[}{]}{}{#1}
\DeclarePairedDelimiterXPP{\expect}[1]{\mathbb{E}^\eta}{[}{]}{}{#1}
\DeclarePairedDelimiterXPP{\expectmu}[1]{\mathbb{E}_\mu}{[}{]}{}{#1}
\DeclarePairedDelimiterXPP{\var}[1]{\mathrm{Var}}{[}{]}{}{#1}
\DeclarePairedDelimiterXPP{\cov}[2]{\mathrm{Cov}}{[}{]}{}{#1, #2}
\newcommand*{\given}[1][]{\nonscript\,#1\vert\nonscript\,\mathopen{}}

\NewDocumentCommand{\E}{m o}{\IfNoValueTF{#2}{\mathcal{#1}}{\mathcal{#1}_{#2}}}
\NewDocumentCommand{\C}{m o}{\IfNoValueTF{#2}{\mathrm{C}_{#1}}{\mathrm{C}_{#1}^{#2}}}
\NewDocumentCommand{\dE}{m o}{\IfNoValueTF{#2}{\dot{\mathcal{#1}}}{\dot{\mathcal{#1}}_{#2}}}
\NewDocumentCommand{\dC}{m o}{\IfNoValueTF{#2}{\dot{\mathrm{C}}_{#1}}{\dot{\mathrm{C}}_{#1}^{#2}}}
\NewDocumentCommand{\EE}{m o}{\mathrm{E}_{#1}^{#2}}
\NewDocumentCommand{\dEE}{m o}{\dot{\mathrm{E}}_{#1}^{#2}}

\begin{document}

\maketitle

\begin{abstract}
Biological neural networks are notoriously hard to model due to their stochastic behavior and high dimensionality. We tackle this problem by constructing a dynamical model of both the expectations and covariances of the fractions of active and refractory neurons in the network's populations. We do so by describing the evolution of the states of individual neurons with a continuous-time Markov chain, from which we formally derive a low-dimensional dynamical system. This is done by solving a moment closure problem in a way that is compatible with the nonlinearity and boundedness of the activation function. Our dynamical system captures the behavior of the high-dimensional stochastic model even in cases where the mean-field approximation fails to do so. Taking into account the second-order moments modifies the solutions that would be obtained with the mean-field approximation, and can lead to the appearance or disappearance of fixed points and limit cycles. We moreover perform numerical experiments where the mean-field approximation leads to periodically oscillating solutions, while the solutions of the second-order model can be interpreted as an average taken over many realizations of the stochastic model. Altogether, our results highlight the importance of including higher moments when studying stochastic networks and deepen our understanding of correlated neuronal activity. 
\end{abstract}

\section{Introduction}

Neuronal activity is intrinsically stochastic. Sources of randomness have been identified at different scales, ranging from spontaneous neurotransmitters release to noise amplification at the network level, and have been shown to contribute to cellular and behavioral trial-to-trial variability \parencite{calvin_synaptic_1968, white_channel_2000, faisal_noise_2008}. Although stochasticity has often been treated as a nuisance, recent studies have highlighted its determining role in neural population coding \parencite{azeredo_da_silveira_geometry_2021} and in many brain functions, such as awareness \parencite{zerlaut_enhanced_2017}, decision making \parencite{lebovich_idiosyncratic_2019}, and working memory \parencite{schneegans_stochastic_2020}.  Moreover, biologically realistic numerical simulations have demonstrated that stochastic networks can support spike-time coding with millisecond precision \parencite{nolte_cortical_2019}, while abstract models have revealed the computational advantage of stochastic spiking neurons over deterministic ones \parencite{maass_noise_2014}.

From a mathematical point of view, incorporating stochasticity into neural network models gives rise to major difficulties, especially when trying to make interpretable predictions about their large-scale dynamical behavior. To address this, a common practice is to perform mean-field approximations \parencite{brunel_dynamics_2000, liley_spatially_2001, ermentrout_mathematical_2010, huang_statistical_2021}, which have long been used in statistical mechanics when studying phase transitions in spin systems \parencite{domb_theory_1960, lebowitz_statistical_1968}. Perhaps the most important example of a mean-field model is the one developed by \textcite{wilson_excitatory_1972}, building on the work of \textcite{beurle_properties_1956} and \textcite{griffith_field_1963, griffith_field_1965}. This model achieved immediate success due to its ability to represent hysteresis phenomena as well as oscillations in biological networks of excitatory and inhibitory neurons. It is still widely used today, sometimes as a firing-rate model \parencite{ermentrout_mathematical_2010, gerstner_time_1995, vogels_neural_2005, keeley_firing_2019}, and has been the starting point of several extensions \parencite{destexhe_wilson-cowan_2009, bressloff_stochastic_2016, cowan_wilsoncowan_2016, chow_before_2020, wilson_evolution_2021}. 

As it is based on a mean-field approximation, the Wilson--Cowan model is restricted to modeling only the average behavior of neural networks, and cannot represent correlations between the states of different neurons. However, correlations have been found to be important in brain activity \parencite{salinas_correlated_2001, schneidman_weak_2006, averbeck_neural_2006, panzeri_structures_2022, azeredo_da_silveira_geometry_2021}, and there has been ongoing effort to generalize the Wilson--Cowan model to take them into account \parencite{chow_before_2020}. An important example of such work is that started by \textcite{buice_field-theoretic_2007, buice_statistical_2009}, later joined by Chow \parencite{buice_systematic_2010, buice_beyond_2013} and followed by \textcite{bressloff_stochastic_2009, bressloff_path-integral_2015} and by \textcite{touboul_finite-size_2011}. These authors used path integral methods to derive a master equation that generalizes the Wilson--Cowan model, and then exploited van Kampen system-size expansions to obtain second-order corrections. This approach has been applied to study the relationship between structure and activity in neural networks \parencite{ocker_linking_2017}.

Although these methods have led to many results, it is not clear that when including refractoriness explicitly in the model, the resulting systems would be appropriate to represent a biological neural network's activity. Indeed, \textcite{painchaud_dynamique_2021} used a different method to obtain a dynamical system analogous to others obtained by \textcite{buice_systematic_2010} and by \textcite{touboul_finite-size_2011}, but which describes the evolution of both active and refractory fractions of neural populations, yet this system can have unbounded solutions which cannot have a biological interpretation. As there is experimental \parencite{berry_refractoriness_1998, avissar_refractoriness_2013} and theoretical \parencite{rule_neural_2019, weistuch_refractory_2021} evidence that refractoriness plays an important role in neuronal activity, it is relevant to retain this aspect in a model that also incorporates correlations between neuronal states.

Here, we propose a new way to generalize the Wilson--Cowan model to explicitly include refractory fractions of neural populations as well as covariances between their activities (including the variance of each activity).\footnote{Throughout the text, unless specified otherwise, we will use the term ``covariances'' to refer both to variances---since a variance is the covariance of a random variable with itself---and to covariances of distinct random variables.} Our main goal is to improve the quality of the mean-field model's prediction of the network's activity while keeping the model relatively simple. To construct our model, the idea is to first build a stochastic process to model the dynamics from the microscopic point of view of individual neurons, and then to use this process to obtain a nonautonomous differential equation that models the same dynamics, but from the macroscopic point of view of neural populations. This nonautonomous differential equation poses a moment closure problem (see \parencite{kuehn_moment_2016} for an extensive discussion), which can be solved while still considering both first- and second-order moments. 

The moment closure framework has not been very popular in computational neuroscience so far, but it has been extensively used in other areas of mathematical biology during the last decades. In particular, it has been used to include covariances in compartmental models in epidemiology \parencite{levin_individuals_1996, keeling_effects_1999, joo_pair_2004, sharkey_deterministic_2008, cator_second-order_2012, kiss_exact_2015} as well as in population dynamics models in ecology \parencite{matsuda_statistical_1992, sato_pathogen_1994}. One of the main advantages of this method is that it reveals precisely where covariances can have an effect on the dynamics. Moreover, it allows a more systematic treatment of the neurons' refractory period, as it has been done by \textcite{painchaud_beyond_2022}. 

The paper is organized as follows. First, we construct a continuous-time Markov chain that describes the evolution of the state of each node in a large network, the goal being to mimic the behavior of biological neurons. This Markov chain is similar to one already proposed (but not extensively studied) by \textcite{cowan_stochastic_1990} and is reminiscent of a process recently studied by \textcite{zarepour_universal_2019}. Then, we split the network into a small set of large populations and derive a nonautonomous system of differential equations that describes the evolution of the expected state of each population. Using a new approximation of the expectation of a sigmoid function, we find a solution to the moment closure problem that involves covariances between population state variables. Finally, we present three examples that demonstrate the impacts of covariances. The first one illustrates how covariances can increase the accuracy when predicting the macroscopic behavior of the Markov chain. In the other two examples, the impact of  including covariances in the dynamical system goes beyond  error reduction: the second-order system reaches a steady state where covariances are nonzero, implying that this steady state is inaccessible to the mean-field model. The paper is complemented by three appendices in which we provide the details of several computations left out of the main text as well as additional figures.

\section{The model}

We seek to describe the dynamics of a large biological neural network from a macroscopic point of view in such a way that the resulting model generalizes Wilson--Cowan's, but also includes correlations between states of neurons. We are also looking for a description that would  depend on statistical properties of neuron parameters rather than on their precise values. To do so, we construct a Markov chain to describe the states of all neurons of the network. Then, we split the network into populations and we derive a dynamical system that describes the evolution of the states of populations. The model constructed in this section is also presented under slightly different assumptions in \parencite{painchaud_beyond_2022}, and it is a special case of a more general model presented in \parencite[Chapter~2]{painchaud_dynamique_2021}.

\subsection{A modelization of a biological neural network}

We consider a network of \(N\) neurons. Links between neurons are described by a random real-valued \(N \times N\) matrix \(W\) that we call the \emph{weight matrix}, and that is defined on a probability space \((H, \mathscr{H}, \mu)\).\footnote{We leave this probability space unspecified as it has no effect on the model, and it can be constructed using the standard method from probability theory. To specify it explicitly, the simplest way is to start from \(\mathbb{R}^{N\times N}\) with its Borel sets \(\mathscr{B}(\mathbb{R}^{N\times N})\) and the law of \(W\), and enlarge it by multiplying it with a space \((\mathbb{R}^N, \mathscr{B}(\mathbb{R}^N), \nu)\) to add a parameter with law \(\nu\) to the model.} An entry \(W_{jk}\) of the weight matrix describes the connection from neuron \(k\) to neuron \(j\); it can be either positive or negative, describing either an excitatory or an inhibitory connection. 

To study the network's dynamics from a macroscopic point of view, we suppose that the neurons are split into a small number \(n\) of distinct subgroups, each consisting of a large number of neurons, which we call \emph{populations}. This splitting is made through a partition \(\mathscr{P}\) of the set \(\{1, \hdots, N\}\) where each \(J \in \mathscr{P}\) corresponds to a population. We assume that the weights \(W_{jk}\) are independent random variables, identically distributed over populations, and assume the same for all parameters that will be introduced in the next section.

\subsection{A description of the evolution of the network's state}

We want the dynamics of our network to model the behavior of biological neurons. The basic behavior we seek to represent is that a neuron fires when it receives sufficient input, and that after firing occurs a short refractory period during which a neuron cannot fire again. To mimic this, we assume that the nodes of our network can take the states described on Figure~\ref{fig.statesandrates}, where the active state corresponds to that of a neuron which is firing. The whole network's state can then be described by an element of the set \(E \defeq \{0,1,i\}^N\).

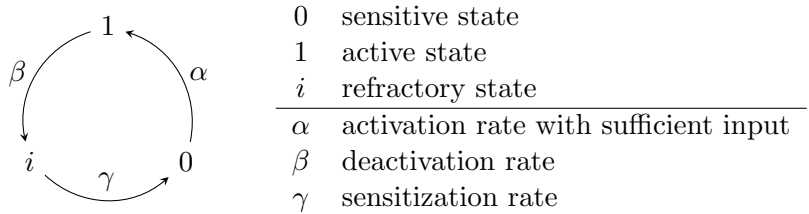
\begin{figure}
\centering
\begin{minipage}[c]{90pt}
    \raggedleft
    \begin{tikzpicture}[x=12mm, y=12mm]
        \node (0) at (-30:1) {\(0\)};
        \node (1) at ( 90:1) {\(1\)};
        \node (i) at (210:1) {\(i\)};
        \draw[-stealth] (0) to[bend right=42] node [midway, right] {\(\alpha\)} (1);
        \draw[-stealth] (1) to[bend right=42] node [midway, left] {\(\beta\)} (i);
        \draw[-stealth] (i) to[bend right=42] node [midway, above] {\(\gamma\)} (0);
    \end{tikzpicture}
\end{minipage}\qquad%
\begin{minipage}[c]{218pt}
    \raggedright
    \begin{tabular}{cl}
        \hline
        \(0\) & sensitive state \\
        \(1\) & active state \\
        \(i\) & refractory state \\
        \hline
        \(\alpha\) & activation rate with sufficient input \\
        \(\beta\) & deactivation rate \\
        \(\gamma\) & sensitization rate \\
        \hline
    \end{tabular}
\end{minipage}
\caption{Neurons' states and allowed transitions between them with corresponding rates. Here, \(i\) denotes the imaginary unit.}
\label{fig.statesandrates}
\end{figure}

To describe the transitions between the states, we consider for each neuron \(j\) three positive random variables \(\alpha_j\), \(\beta_j\) and \(\gamma_j\) and a real-valued one \(\theta_j\), which are all independent and defined on \((H, \mathscr{H}, \mu)\). The possible transitions are described on Figure~\ref{fig.statesandrates}. The parameters \(\beta_j\) and \(\gamma_j\) are the rates at which neuron \(j\) can make the transitions \(1 \mapsto i\) and \(i \mapsto 0\). The activation rate is described by a function \(a_j\colon H \times E \to (0,\infty)\) given by
\[
a_j(\eta, x) \defeq \alpha_j(\eta) \charf{T_j(x)}(\eta),
\]
where \(\charf{T_j(x)}\) is the indicator function of the set
\[
T_j(x) \defeq \Bigl\{ \eta \in H : \sum_{k=1}^N W_{jk}(\eta) \Re x_k + Q_J > \theta_j(\eta) \Bigr\},
\]
\(Q_J \in \mathbb{R}\) being a deterministic external input received by the population \(J \in \mathscr{P}\) to which belongs neuron \(j\). Thus, \(j\) can activate at rate \(\alpha_j\) when its input exceeds its threshold \(\theta_j\).

This description of the microscopic behavior is the intuitive description of a continuous-time Markov chain, which we now define properly. We introduce for each choice of parameters \(\eta \in H\) a matrix \(M^\eta = \{m^\eta(x,y) : x, y \in E\}\) with entries
\[
m^\eta(x,y) \defeq \sum_{j=1}^N m_j^\eta(x,y) \smash{\prod_{\substack{k=1 \\ k\neq j}}^N} \delta_{x_ky_k}
\]
where
\[
\begin{aligned}
m_j^\eta(x,y) \defeq a_j(\eta, x) & (1 - \abs{x_j}) \bigl( \Re y_j - (1 - \abs{y_j}) \bigr) \\
    {} + \beta_j(\eta) & \Re(x_j) (\Im y_j - \Re y_j) \\
    {} + \gamma_j(\eta) & \Im(x_j) \bigl( (1 - \abs{y_j}) - \Im y_j \bigr)
\end{aligned}
\]
and \(\delta_{xy}\) is a Kronecker delta. To make sense of \(m^\eta_j(x,y)\), recall that \(x_j \in \{0,1,i\}\), so that exactly one of \(\Re x_j\), \(\Im x_j\) and \(1 - \abs{x_j}\) is \(1\) while the other two are \(0\), which implies that \(m^\eta_j(x,y)\) is always \(\pm a_j(\eta,x)\), \(\pm\beta_j(\eta)\) or \(\pm\gamma_j(\eta)\). A simple calculation then shows that the matrix \(M^\eta\) is the generator of a continuous-time Markov chain. Thus, it follows from the Kolmogorov extension theorem (see e.g.~\parencite{doob_stochastic_1990, norris_markov_1997} for details) that there exists a probability measure \(\mathbb{P}^\eta\) on \((\Omega,\mathscr{F}) \defeq (E, 2^E)^{[0,\infty)}\) such that for any \(x, y \in E\), as \(\Delta t \downarrow 0\),
\[
\bprob{X_{t+\Delta t}=y \given X_t=x} = \delta_{xy} + m^\eta(x,y) \Delta t + o(\Delta t),
\]
where \(\{X_t\}_{t\geq 0}\) is the coordinate mapping process \(X_t(\omega) \defeq \omega(t)\) on \((\Omega,\mathscr{F})\). In particular, if \(x \in E\) has \(x_j = 0\), then
\begin{subequations}
\label{eq.transitionprobabilities}
\begin{align}
\SwapAboveDisplaySkip
\bprob{X_{t+\Delta t}^j=1 \given X_t=x} & = a_j(\eta,x) \Delta t + o(\Delta t),
\intertext{while in general}
\bprob{X_{t+\Delta t}^j=i \given X_t^j=1} & = \beta_j(\eta) \Delta t + o(\Delta t), \\
\bprob{X_{t+\Delta t}^j=0 \given X_t^j=i} & = \gamma_j(\eta) \Delta t + o(\Delta t),
\end{align}
\end{subequations}
and the forbidden transitions \(0 \mapsto i\), \(i \mapsto 1\) and \(1 \mapsto 0\) all have \(o(\Delta t)\) rates. Hence, the stochastic process \(\{X_t\}_{t\geq 0}\) describes the state of the network with respect to time as described at the beginning of the section.

\subsection{A macroscopic approximation of the Markov chain's behavior}

The above Markov chain completely describes the dynamics of the network from a microscopic point of view. Now, we want to find an approximation of this Markov chain to describe the macroscopic behavior of the network. To make this goal more precise, we introduce for each population \(J \in \mathscr{P}\) the processes
\beq{eq.defARS}
A_t^J \defeq \frac{1}{\abs{J}} \sum_{j\in J} \Re X_t^j, \quad
R_t^J \defeq \frac{1}{\abs{J}} \sum_{j\in J} \Im X_t^j, \quad
S_t^J \defeq \frac{1}{\abs{J}} \sum_{j\in J} \bigl( 1 - \abs{X_t^j} \bigr),
\eeq
which are the fractions of neurons in population \(J\) that are in the active, refractory and sensitive state respectively. Our goal is to find a dynamical system to describe the evolution of
\beq{eq.defEARS}
\E{A}[J](t) \defeq \expect{A_t^J}, \qquad
\E{R}[J](t) \defeq \expect{R_t^J}, \qquad
\E{S}[J](t) \defeq \expect{S_t^J},
\eeq
where \(\mathbb{E}^\eta\) denotes the expectation on \((\Omega,\mathscr{F},\mathbb{P}^\eta)\).

A first step towards an understanding of the evolution of these expected fractions of populations is a description of the evolution of the probabilities
\begin{alignat*}{2}
p_j(t) & \defeq \bprob{X_t^j = 1} && = \expect[\big]{\Re X_t^j}, \\
r_j(t) & \defeq \bprob{X_t^j = i} && = \expect[\big]{\Im X_t^j}, \\
q_j(t) & \defeq \bprob{X_t^j = 0} && = \expect[\big]{1 - \abs{X_t^j}}.
\end{alignat*}
By considering these probabilities at a time \(t + \Delta t\) and conditioning over all events \(\{X_t = x\}\) for \(x \in E\), it is not hard to verify using the transition probabilities given in \cref{eq.transitionprobabilities} that
\begin{subequations}
\label{eq.microsys}
\begin{align}
\label{eq.dpj}
    \dot{p}_j(t) & = - \beta_j(\eta) p_j(t) + \expect[\big]{a_j(\eta, X_t) \bigl( 1 - \abs{X_t^j} \bigr)}, \\
\label{eq.dqj}
    \dot{q}_j(t) & = - \expect[\big]{a_j(\eta, X_t) \bigl( 1 - \abs{X_t^j} \bigr)} + \gamma_j(\eta) r_j(t), \\
\label{eq.drj}
    \dot{r}_j(t) & = - \gamma_j(\eta) r_j(t) + \beta_j(\eta) p_j(t).
\end{align}
\end{subequations}
The details are given in Appendix~\ref{apx.DSmicro1}. We use these differential equations as the starting point to describe the evolution of the expected fractions of populations defined in \cref{eq.defEARS}. Remark that this microscopic model has essentially the same form as models based on firing rates, except for the presence of the factor \(1 - \abs{X_t^j}\) in activation terms.

To obtain expressions for the derivatives of \(\E{A}[J]\), \(\E{R}[J]\) and \(\E{S}[J]\), the idea is to average the derivatives given in \cref{eq.microsys} over \(J\) using the linearity of expectation and derivatives. To avoid overcomplicating the discussion, we give here only the main ideas of the arguments that lead us to the differential equations. The complete details, including the limits of the approximations that are made, are given in Appendix~\ref{apx.DS}. To obtain an expression for \(\dE{A}[J]\), we start from the average
\[
\dE{A}[J](t) = \expect[\bigg]{\frac{1}{\abs{J}} \sum_{j \in J} \Bigl( - \beta_j(\eta) \Re X_t^j + a_j(\eta,X_t) \bigl( 1 - \abs{X_t^j} \bigr) \Bigr)}.
\]
Assuming that each population is large, the law of large numbers now motivates to approximate the averages of transition rates over populations by their expectations. This yields the approximation
\(
\frac{1}{\abs{J}} \sum_{j\in J} \beta_j(\eta) \Re X_t^j \approx \beta_J A_t^J
\)
where \(\beta_J \defeq \expectmu{\beta_j}\) for \(j \in J\), \(\mathbb{E}_\mu\) being the expectation on \((H, \mathscr{H}, \mu)\). The activation rates are handled in the same way, but since these rates are proportional to step functions of the difference between neuronal inputs and thresholds, their expectations are proportional to the cumulative distribution functions \(F_{\theta_J}\) of the thresholds in population \(J\) evaluated at a population-averaged input:
\[
\frac{1}{\abs{J}} \sum_{j\in J} a_j(\eta, X_t) \bigl( 1 - \abs{X_t^j} \bigr) \approx \alpha_J F_{\theta_J}(B_t^J) S_t^J
\]
where \(\alpha_J \defeq \expectmu{\alpha_j}\) and
\[
B_t^J \defeq \sum_{K\in\mathscr{P}} c_{JK} A_t^K + Q_J
\quad\text{with}\quad
c_{JK} \defeq \abs{K} \expectmu{W_{jk}}
\]
for \(j \in J\) and \(k \in K\). 

This method results in approximate expressions for the derivatives of \(\E{A}[J]\), \(\E{R}[J]\) and \(\E{S}[J]\) when populations are large. It ultimately leads to model the macroscopic dynamics of the network by the differential equations
\begin{subequations}
\label{eq.DSopen}
\begin{align}
\SwapAboveDisplaySkip
\label{eq.dAJ}
\dE{A}[J](t) & = - \beta_J \E{A}[J](t) + \alpha_J \expect{F_{\theta_J}(B_t^J) S_t^J}, \\
\label{eq.dRJ}
\dE{R}[J](t) & = - \gamma_J \E{R}[J](t) + \beta_J \E{A}[J](t), \\
\label{eq.dSJ}
\dE{S}[J](t) & = - \alpha_J \expect{F_{\theta_J}(B_t^J) S_t^J} + \gamma_J \E{R}[J](t),
\end{align}
\end{subequations}
where \(\gamma_J \defeq \expectmu{\gamma_j}\) for \(j \in J\). For each population, one of these three equations is redundant since \(A_t^J + R_t^J + S_t^J \equiv 1\) for all \(t\). Thus, if the network has \(n\) populations, \cref{eq.DSopen} corresponds to \(2n\) independent differential equations. In the following, we use the active and refractory fractions as the independent variables, and always see sensitive fractions simply as a function of them.

A crucial aspect of \cref{eq.DSopen} is that it is not autonomous: there is an explicit time dependence in the expectation \(\expect{F_{\theta_J}(B_t^J) S_t^J}\), which is additionally an unknown function. Further assumptions or approximations are required to obtain an autonomous system---this corresponds to a moment closure problem \parencite{kuehn_moment_2016}. The simplest way to close the system would be to use the mean-field approximation \(\expect{F_{\theta_J}(B_t^J) S_t^J} \approx F_{\theta_J}\bigl( \E{B}[J](t) \bigr) \E{S}[J](t)\), where \(\E{B}[J](t) \defeq \expect{B_t^J}\). This results in the dynamical system
\begin{subequations}
\label{eq.meanfield}
\begin{align}
\SwapAboveDisplaySkip
\dE{A}[J] & = - \beta_J \E{A}[J] + \alpha_J F_{\theta_J}(\E{B}[J]) \E{S}[J], \\
\dE{R}[J] & = - \gamma_J \E{R}[J] + \beta_J \E{A}[J],
\end{align}
\end{subequations}
which is closed since for each population \(J\), \(\E{S}[J]\) and \(\E{B}[J]\) are functions of expectations of active and refractory fractions. Since the approximation used to close the system does not depend on the specific choice of parameters \(\eta \in H\), any explicit dependence on \(\eta\) disappeared. Thus, the macroscopic dynamics predicted by the model only depend on the averages of parameters over populations. This will be the case for any macroscopic model obtained from a closure of the system given in \cref{eq.DSopen}, as long as the approximations used to close the system do not depend explicitly on \(\eta\). 

The mean-field system given in \cref{eq.meanfield} has been studied in detail by \textcite{painchaud_beyond_2022} where it is shown to be closely related to Wilson--Cowan's equations \parencite{wilson_excitatory_1972}: if the refractory fractions are forced to their equilibrium solutions \(\E{R}[J] = \frac{\beta_J}{\gamma_J} \E{A}[J]\), the resulting system is equivalent to Wilson--Cowan's. In the next section, we go further and present a solution to the moment closure problem posed by \cref{eq.DSopen} that includes second-order moments, with the goal of improving the mean-field model's predictions of the expectations of the active, refractory and sensitive fractions of each population.

\section{A second-order solution to the moment closure problem}

We are looking for a solution to the moment closure problem posed by \cref{eq.DSopen} that includes covariances between fractions of populations. The first step is to find how the second moments of the active and refractory fractions of populations evolve in time. To simplify notation in the following, we define
\[
\C{YZ}[JK](t)
    \defeq \cov{Y_t^J}{Z_t^K}
    = \expect{Y_t^J Z_t^k} - \expect{Y_t^J} \expect{Z_t^K},
\]
where \(Y\) and \(Z\) stand for either \(A\), \(R\), \(S\) or \(B\), and \(J, K \in \mathscr{P}\). Then, with the method we used above to find \cref{eq.DSopen}, it can be shown (and this is done in Appendix~\ref{apx.DS}) that
\begin{subequations}
\label{eq.DScovs}
\begin{align}
\dC{AA}[JK](t)
    & = - (\beta_J + \beta_K) \C{AA}[JK](t) + \alpha_K \cov{A_t^J}{F_{\theta_K}(B_t^K) S_t^K} + \alpha_J \cov{A_t^K}{F_{\theta_J}(B_t^J) S_t^J}, \\
\dC{RR}[JK](t)
    & = - (\gamma_J + \gamma_K) \C{RR}[JK](t) + \beta_K \C{AR}[KJ](t) + \beta_J \C{AR}[JK](t), \\
\dC{AR}[JK](t)
    & = - (\beta_J + \gamma_K) \C{AR}[JK](t) + \beta_K \C{AA}[JK](t) + \alpha_J \cov{R_t^K}{F_{\theta_J}(B_t^J) S_t^J}.
\end{align}
\end{subequations}
Thus, to obtain a closed dynamical system involving only first and second moments of active and refractory fractions, we need to find approximations to expectations of the forms \(\expect{F_{\theta_J}(B_t^J) S_t^J}\),  \(\expect{F_{\theta_J}(B_t^J) S_t^J A_t^K}\) and \(\expect{F_{\theta_J}(B_t^J) S_t^J R_t^K}\) in terms of expectations and covariances of active and refractory fractions of populations. Remark that, since sensitive fractions and inputs are expressed in terms of active and refractory fractions by linear relations, the bilinearity of covariance implies that any covariance of the form \(\C{YZ}[JK]\) with \(Y\) and \(Z\) being \(A\), \(R\), \(S\) or \(B\) can be expressed using active and refractory fractions only.

In order to lighten notation, the time dependence will be kept implicit in the remainder of the section. For instance, we will write \(B^J\) instead of \(B_t^J\).

\subsection{The naive approach and the need to go further}

The simplest solution to this moment closure problem is to assume that \(F_{\theta_J}\) is smooth enough and to use local approximations given by Taylor expansions to write the problematic expectations in terms of central moments of the dynamical variables. Then, all central moments of order higher than 2 can be neglected to close the system. Indeed, expanding \(F_{\theta_J}\) around the expectation \(\E{B}[J]\) of \(B^J\) and assuming that \(F_{\theta_J}\) is regular enough so that the expectation can be distributed in the series,
\[
\expect{F_{\theta_J}(B^J) S^J}
    = \sum_{k=0}^\infty \frac{1}{k!} F_{\theta_J}^{(k)}(\E{B}[J]) \E{S}[J] \expect[\big]{(B^J - \E{B}[J])^k}
    + \sum_{k=0}^\infty \frac{1}{k!} F_{\theta_J}^{(k)}(\E{B}[J]) \expect[\big]{(S^J - \E{S}[J]) (B^J - \E{B}[J])^k}.
\]
Now, all terms with \(k \geq 3\) in the first series or with \(k \geq 2\) in the second series are proportional to central moments of order at least 3. Neglecting these terms and keeping only the nonzero remaining ones, we find the approximation
\beq{eq.approxTaylorEFS}
\expect{F_{\theta_J}(B^J) S^J}
    \approx F_{\theta_J}(\E{B}[J]) \E{S}[J] + \frac{1}{2} F_{\theta_J}''(\E{B}[J]) \E{S}[J] \C{BB}[JJ] + F_{\theta_J}'(\E{B}[J]) \C{SB}[JJ].
\eeq
As \(B^J\) and \(S^J\) are related to active and refractory fractions of populations by linear relations, using this approximation in \cref{eq.dAJ} allows to write the derivative \(\dE{A}[J]\) in terms of \(\E{A}[J]\)'s and \(\E{R}[J]\)'s only.

The same method can be used on the covariances in the differential equations~\eqref{eq.DScovs}, and leads to
\beq{eq.approxTaylorCAFS}
\cov{A^J}{F_{\theta_K}(B^K) S^K}
    \approx F_{\theta_K}(\E{B}[K]) \C{AS}[JK] + F_{\theta_K}'(\E{B}[K]) \C{AB}[JK] \E{S}[K]
\eeq
as well as a similar one in which \(A\) is replaced with \(R\). 

Using the approximations given in \cref{eq.approxTaylorEFS,eq.approxTaylorCAFS} in the system given by \cref{eq.DSopen,eq.DScovs} leads to the system
\begin{subequations}
\label{eq.DSTaylor}
\begin{align}
\dE{A}[J]
    & = - \beta_J \E{A}[J] + \alpha_J F_{\theta_J}(\E{B}[J]) \E{S}[J] + \alpha_J F_{\theta_J}'(\E{B}[J]) \C{SB}[JJ] + \frac{\alpha_J}{2} F_{\theta_J}''(\E{B}[J]) \E{S}[J] \C{BB}[JJ], \\
\dE{R}[J]
    & = - \gamma_J \E{R}[J] + \beta_J \E{A}[J], \\
\begin{split}
\dC{AA}[JK]
    & = - (\beta_J + \beta_K) \C{AA}[JK] + \alpha_K F_{\theta_K}(\E{B}[K]) \C{AS}[JK] + \alpha_K F_{\theta_K}'(\E{B}[K]) \E{S}[K] \C{AB}[JK] \\
    &\hspace*{44mm} + \alpha_J F_{\theta_J}(\E{B}[J]) \C{AS}[KJ] + \alpha_J F_{\theta_J}'(\E{B}[J]) \E{S}[J] \C{AB}[KJ],
\end{split} \\
\dC{RR}[JK]
    & = - (\gamma_J + \gamma_K) \C{RR}[JK] + \beta_K \C{AR}[KJ] + \beta_J \C{AR}[JK], \\
\dC{AR}[JK]
    & = - (\beta_J + \gamma_K) \C{AR}[JK] + \beta_K \C{AA}[JK] + \alpha_J F_{\theta_J}(\E{B}[J]) \C{RS}[KJ] + \alpha_J F_{\theta_J}'(\E{B}[J]) \E{S}[J] \C{RB}[KJ],
\end{align}
\end{subequations}
which is a closed dynamical system since the linear relations \(S^J = 1 - A^J - R^J\) and \(B^J = \sum_{K\in\mathscr{P}} c_{JK} A^K\) allow to write all covariances in terms of covariances between active and refractory fractions of populations. 

This system has already been studied in detail by \textcite{painchaud_dynamique_2021}, and is reminiscent of other ones that appear in the literature. In particular, neglecting the refractory fractions as well as the factors corresponding to sensitive fractions in the activation terms, the system becomes the infinite-size system of \textcite{touboul_finite-size_2011}, whose one-population case had previously been derived by \textcite{buice_systematic_2010}. 

\begin{figure}
\includegraphics[width=\linewidth]{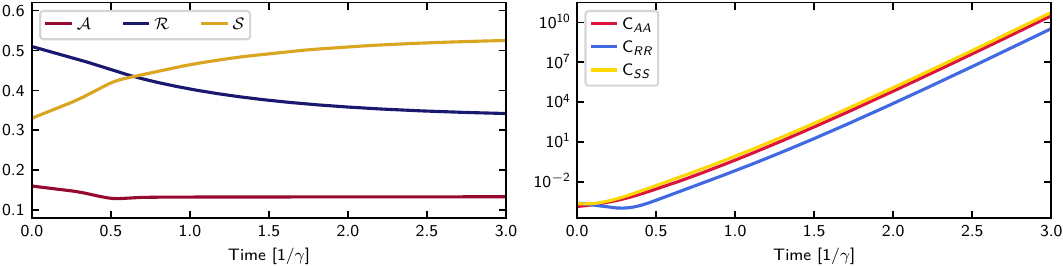}
\caption{Solution of the dynamical system from \cref{eq.DSTaylor} with exponentially diverging variances, an inadmissible behavior. The parameters and initial state were chosen as in the example of Section~\ref{sec.example1}.}
\label{fig.explosions}
\end{figure}

Nevertheless, this solution to the moment closure problem does not seem to be the most appropriate one, at least when refractory fractions of populations are to be included. Indeed, \textcite[Section~4.4]{painchaud_dynamique_2021} found examples in which the solutions of \cref{eq.DSTaylor} have components that blow up. An example is given on Figure~\ref{fig.explosions}, which depicts a solution whose components corresponding to variances diverge exponentially. In such cases, it is clear that the model cannot be interpreted biologically. This failure could be explained by the fact that the Taylor approximation from which the system is derived should only hold locally, when covariances are small. As there is no guarantee that the dynamical system will keep small covariances small, one could expect this to be a problem, especially since the right-hand sides of \cref{eq.approxTaylorEFS,eq.approxTaylorCAFS} are unbounded (with respect to covariances), contrary to the functions they seek to approximate. It might be possible to solve this problem by including higher order terms in the dynamical system, in order to improve the accuracy of the approximations made in \cref{eq.approxTaylorEFS,eq.approxTaylorCAFS}, but this would also have the effect to enlarge the system's dimension. With the goal of improving the approximation without enlarging the dimension, this leads us to propose a different solution to the moment closure problem, which we build around a function that approximates the expectation of a sigmoid function of a random variable in terms of its first and second moments.

\subsection{Approximation of a sigmoid function's expectation}

The difficulty in approximating the expectations \(\expect{F_{\theta_J}(B^J) S^J}\) or \(\expect{F_{\theta_J}(B^J) S^J A^K}\) comes, to a large extent, from the cumulative distribution function \(F_{\theta_J}\). As a first step, we construct an approximation of the expectation of \(F_{\theta_J}(B^J)\) only. To do so, we assume that in each population, the thresholds follow a symmetric, unimodal distribution, so that \(F_{\theta_J}\) has a sigmoid shape with inflection point at the mean \(\theta_J\).

To construct our approximation, we ask a simple question: if the random variable \(B^J\) follows a symmetric distribution with mean \(\E{B}[J]\) and a small variance \(\C{BB}[JJ]\), how should the expectation of \(F_{\theta_J}(B^J)\) compare to \(F_{\theta_J}(\E{B}[J])\)? It is not hard to find a qualitative answer. Indeed, when \(\E{B}[J] < \theta_J\), the fact that \(F_{\theta_J}\) increases faster to the right of \(\E{B}[J]\) than it decreases to its left has the effect that \(F_{\theta_J}(\E{B}[J])\) must underestimate the actual expectation of \(F_{\theta_J}(B^J)\). This effect is reversed when \(\E{B}[J] > \theta_J\). This is sketched on Figure~\ref{fig.approxEF}. 

\begin{figure}
\includegraphics[width=\linewidth]{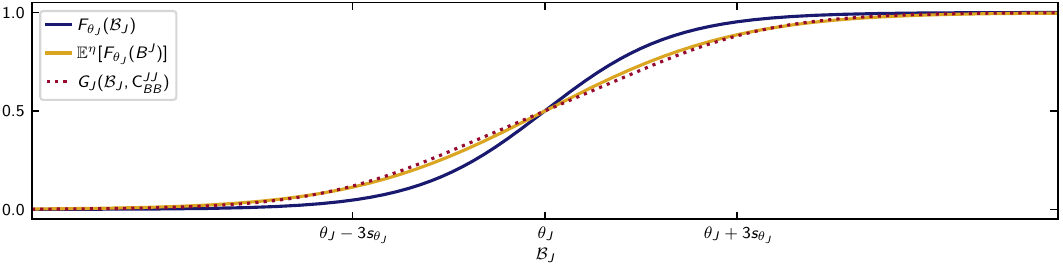}
\caption{Example of comparison between \(F_{\theta_J}(\E{B}[J])\), the actual expectation \(\expect{F_{\theta_J}(B^J)}\), and our approximation \(G_J(\E{B}[J], \C{BB}[JJ])\) for a fixed variance \(\C{BB}[JJ]\), assuming the thresholds follow a logistic distribution with mean \(\theta_J\) and scaling factor \(s_{\theta_J}\). To compute \(\expect{F_{\theta_J}(B^J)}\), we assumed as a heuristic device that \(B^J\) follows a logistic distribution with mean \(\E{B}[J]\) and scaling factor \(s_{\theta_J}\) as well, but this choice was made for illustration purposes only.}
\label{fig.approxEF}
\end{figure}

Observe from Figure~\ref{fig.approxEF} that the deformation of \(F_{\theta_J}(\E{B}[J])\) into \(\expect{F_{\theta_J}(B^J)}\) resembles an increase in the standard deviation of the distribution of which \(F_{\theta_J}\) is the cumulative distribution function. If \(\sigma_{\theta_J}\) denotes the standard deviation of a threshold \(\theta_j\) of population \(J\), then increasing the standard deviation of \(\theta_j\) by a factor \(1 + a\) for some \(a > 0\) corresponds to changing \(\theta_j\) to \(\theta_j(1 + a) - \theta_J a\). The cumulative distribution functions of the two cases can be related as
\[
\bprob{\theta_j (1 + a) - \theta_J a \leq x}
    = \bprob[\Big]{\theta_j \leq \frac{x + \theta_J a}{1 + a}}.
\]
Therefore, our previous observations suggest to look for an approximation of the form
\beq{eq.approxEF}
\expect{F_{\theta_J}(B^J)} \approx G_J\bigl( \E{B}[J], \C{BB}[JJ] \bigr)
\eeq
with
\beq{eq.defGJ.main}
G_J(b,v) \defeq F_{\theta_J}\Bigl( \frac{b + \theta_J g_J(b,v)}{1 + g_J(b,v)} \Bigr),
\eeq
for some \(g_J\colon \mathbb{R} \times [0,\infty) \to [0,\infty)\) to be determined.

We can find a candidate for \(g_J\) based on what \(G_J\) should be for small variances. Indeed, if the variance \(\C{BB}[JJ]\) of \(B^J\) is small, then a second-order Taylor approximation of \(F_{\theta_J}\) should provide a reasonable approximation of \(\expect{F_{\theta_J}(B^J)}\), so that
\beq{eq.approxTaylorEF}
\expect{F_{\theta_J}(B^J)} \approx F_{\theta_J}(\E{B}[J]) + \frac{1}{2} F_{\theta_J}''(\E{B}[J]) \C{BB}[JJ].
\eeq
For the approximation given by \cref{eq.approxEF} to be consistent with this, it must be that
\[
G_J(b,0) = F_{\theta_J}(b)
\quad\text{and}\quad
\partial_2G_J(b,0) = \frac{1}{2} F_{\theta_J}''(b)
\]
for any \(b \in \mathbb{R}\), where \(\partial_2 G_J\) denotes the partial derivative of \(G_J\) with respect to its second argument. Computing this derivative and using the two conditions yields
\[
\frac{1}{2} F_{\theta_J}''(b) = F_{\theta_J}'(b) (\theta_J - b) \partial_2 g_J(b,0),
\]
leading us to define
\beq{eq.defgJ.main}
g_J(b,v) \defeq \frac{v}{2(\theta_J - b)} \frac{F_{\theta_J}''(b)}{F_{\theta_J}'(b)}
\eeq
for \(b \neq \theta_J\), assuming that \(F_{\theta_J}' > 0\). As long as the distribution of the thresholds in population \(J\) is unimodal and symmetric, \(F_{\theta_J}''(b)\) always has the same sign as \(\theta_J - b\), and in particular \(F_{\theta_J}''(\theta_J) = 0\). Hence, \(g_J\) can be continuously extended to \(b = \theta_J\) by definition of the derivative (provided \(F_{\theta_J}'''(\theta_J)\) exists), and it is nonnegative on \(\mathbb{R} \times [0,\infty)\).

For the approximation given by \cref{eq.approxEF} with \(g_J\) defined as above to make sense, we must ensure that \(G_J\) has several properties. First, \(G_J\bigl( \E{B}[J], \C{BB}[JJ] \bigr)\) should still be consistent with \(\expect{F_{\theta_J}(B^J)}\) in extreme cases: if \(B^J\) has zero variance, then it should be possible to replace it with its mean, while if it has infinite variance then \(F_{\theta_J}(B^J)\) should be \(0\) or \(1\) with \(\nicefrac{1}{2}\)--\(\nicefrac{1}{2}\) probabilities. Additionally, if the expected input \(\E{B}[J]\) goes to \(+\infty\) or \(-\infty\), then \(F_{\theta_J}(B^J)\) should go to \(1\) or \(0\) regardless of the variance. Moreover, since \(F_{\theta_J}\) is increasing, \(G_J\bigl( \E{B}[J], \C{BB}[JJ] \bigr)\) should always increase with \(\E{B}[J]\) if the variance is fixed. 

Adding some assumptions on the thresholds' distribution allows to prove that \(G_J\) indeed has all of the properties enumerated above, in addition to being consistent with the Taylor approximation of \(F_{\theta_J}\) given in \cref{eq.approxTaylorEF}.

\begin{theorem}
\label{thm.GJ.main}
Suppose that the thresholds in population \(J\) follow a unimodal and symmetric distribution with mean \(\theta_J\) and cumulative distribution function \(F_{\theta_J}\). Let \(g_J\) and \(G_J\) be defined by \cref{eq.defgJ.main,eq.defGJ.main} respectively. Suppose that
\begin{enumerate}[label=\enumlabelformat{\arabic*}]
    \item \(F_{\theta_J}\) is \(\mathscr{C}^4\) on \(\mathbb{R}\);
    \item \(F_{\theta_J}' > 0\);
    \item \(g_J(\cdot, 1)\) is bounded on \(\mathbb{R}\);
    \item\label{thm.GJ.main.conditionDgJ}
        \(\forall b \in \mathbb{R}\), \(g_J(b,1) + (\theta_J - b) \partial_1 g_J(b,1) \geq 0\).
\end{enumerate}
Then \(G_J\) satisfies the following conditions.
\begin{enumerate}
    \item \(G_J\) is \(\mathscr{C}^1\) on \(\mathbb{R} \times [0,\infty)\);
    \item \(G_J(\cdot, 0) = F_{\theta_J}\) and \(\forall b \in \mathbb{R}, G_J(b, v) \to \nicefrac{1}{2}\) as \(v \to \infty\);
    \item \(\forall v \geq 0, G_J(b,v) \to 0\) as \(b \to -\infty\) and \(G_J(b,v) \to 1\) as \(b \to \infty\);
    \item \(\forall v \geq 0\), \(G_J(\cdot,v)\) is increasing;
    \item \(\partial_2G_J(\cdot,0) = \frac{1}{2} F_{\theta_J}''\).
\end{enumerate}
\end{theorem}

The proof is given in Appendix~\ref{apx.proofs}. Assumption~\ref{thm.GJ.main.conditionDgJ} does not have a simple interpretation, but we use it to prove that the maps \(G_J(\cdot, v)\) are increasing, and it is fairly easy to verify for given distributions. In particular, we show in Appendix~\ref{apx.proofs} that if the thresholds follow either a normal or a logistic distribution, then all assumptions hold. An example of the function \(G_J(\cdot, v)\) for a fixed variance \(v\) is shown on Figure~\ref{fig.approxEF}.

\subsection{The moment closure}

Approximations to expectations of the form \(\expect{F_{\theta_J}(B^J) S^J}\) and \(\expect{F_{\theta_J}(B^J) S^J A^K}\) can be constructed using the function \(G_J\) from the previous section. We start by giving an approximation of \(\expect{F_{\theta_J}(B^J) S^J}\) in terms of the expectations \(\E{S}[J]\) and \(\E{B}[J]\) as well as the covariance \(\C{SB}[JJ]\) and the variance \(\C{BB}[JJ]\). As in the simpler case of the expectation of \(F_{\theta_J}(B^J)\), the approximation of \(\expect{F_{\theta_J}(B^J) S^J}\) must meet some requirements. First, the difference between the approximation and \(\expect{F_{\theta_J}(B^J)} \E{S}[J]\) should approximate the covariance between \(F_{\theta_J}(B^J)\) and \(S^J\), so its sign should always follow that of the covariance \(\C{SB}[JJ]\) because \(F_{\theta_J}\) is increasing. Then, if \(\E{B}[J]\) goes to \(+\infty\) or \(-\infty\), \(F_{\theta_J}(B^J)\) should go to either \(1\) or \(0\), so \(\expect{F_{\theta_J}(B^J) S^J}\) should go to either \(\E{S}[J]\) or \(0\). Finally, the approximation we obtained in \cref{eq.approxTaylorEFS} from a Taylor approximation of \(F_{\theta_J}\), 
\[
\expect{F_{\theta_J}(B^J) S^J}
    \approx F_{\theta_J}(\E{B}[J]) \E{S}[J]
    + \frac{1}{2} F_{\theta_J}''(\E{B}[J]) \E{S}[J] \C{BB}[JJ]
    + F_{\theta_J}'(\E{B}[J]) \C{SB}[JJ],
\]
should be valid when covariances are small. This gives conditions on the value of the approximation and of its partial derivatives when \(\C{SB}[JJ] = \C{BB}[JJ] = 0\).

Now, the function \(G_J\) defined in \cref{eq.defGJ.main} can be used to construct the desired approximation. The following result follows from Theorem~\ref{thm.GJ.main} and is proven in Appendix~\ref{apx.proofs}.

\begin{corollary}
Suppose that all assumptions of Theorem~\ref{thm.GJ.main} hold. Then \(f\colon (0,\infty) \times \mathbb{R}^3 \to (0,\infty)\) defined by
\[
f(s,b,c,v) \defeq s G_J\Bigl( b + \frac{c}{s}, v \Bigr)
\]
satisfies the following conditions for any \(s > 0\), \(b,c \in \mathbb{R}\) and \(v \geq 0\).
\begin{enumerate}
    \item \(f\) is \(\mathscr{C}^1\) on \((0,\infty) \times \mathbb{R}^3\);
    \item \(f(s,b,c,v) \gtreqless sG_J(b,v)\) when \(c \gtreqless 0\), and in particular \(f(s,b,0,v) = sG_J(b,v)\) and \(f(s,b,0,0) = sF_{\theta_J}(b)\);
    \item \(f(s,b,c,v) \to 0\) as \(b\to-\infty\) and \(f(s,b,c,v) \to s\) as \(b \to \infty\);
    \item \(\partial_3f(s,b,0,0) = F_{\theta_J}'(b)\);
    \item \(\partial_4f(s,b,0,0) = \frac{1}{2} s F_{\theta_J}''(b)\).
\end{enumerate}
\end{corollary}

These results motivate the approximation
\beq{eq.approxEFS}
\expect{F_{\theta_J}(B^J) S^J} \approx \E{S}[J] G_J\Bigl( \E{B}[J] + \frac{\C{SB}[JJ]}{\E{S}[J]}, \C{BB}[JJ] \Bigr).
\eeq

Approximations to \(\expect{F_{\theta_J}(B^J) S^J A^K}\) and \(\expect{F_{\theta_J}(B^J) S^J R^K}\) can be constructed in a similar way. Now, the only obvious conditions that the approximation should satisfy are those related to the Taylor approximation of \(F_{\theta_J}\). Indeed, neglecting third central moments as before,
\begin{multline*}
\expect{F_{\theta_J}(B^J) S^J A^K}
    \approx \bigl( \E{A}[K]\E{S}[J] + \C{AS}[KJ] \bigr) F_{\theta_J}(\E{B}[J]) \\
    + \bigl( \E{A}[K] \C{SB}[JJ] + \E{S}[J] \C{AB}[KJ] \bigr) F_{\theta_J}'(\E{B}[J])
    + \frac{1}{2} \E{A}[K] \E{S}[J] F_{\theta_J}''(\E{B}[J]) \C{BB}[JJ],
\end{multline*}
which should hold when covariances are small. The appropriate approximation is defined in the following result, which follows from Theorem~\ref{thm.GJ.main} and is proven in Appendix~\ref{apx.proofs}.

\begin{corollary}
Suppose that all assumptions of Theorem~\ref{thm.GJ.main} hold. Then \(f \colon (0,\infty)^2 \times \mathbb{R}^5 \to \mathbb{R}\) defined by
\[
f(x,s,b,c_1,c_2,c_3,v) \defeq (xs + c_1) G_J\Bigl( b + \frac{c_2}{x} + \frac{c_3}{s}, v \Bigr)
\]
satisfies the following conditions for any \(x,s > 0\), \(b \in \mathbb{R}\) and \(v \geq 0\).
\begin{enumerate}
    \item \(f\) is \(\mathscr{C}^1\) on \((0,\infty)^2 \times \mathbb{R}^5\);
    \item \(f(x,s,b,0,0,0,v) = xs G_J(b,v)\), and in particular \(f(x,s,b,0,0,0,0) = xs F_{\theta_J}(b)\);
    \item \(\partial_4 f(x,s,b,0,0,0,0) = F_{\theta_J}(b)\);
    \item \(\partial_5 f(x,s,b,0,0,0,0) = s F_{\theta_J}'(b)\);
    \item \(\partial_6 f(x,s,b,0,0,0,0) = x F_{\theta_J}'(b)\);
    \item \(\partial_7 f(x,s,b,0,0,0,0) = \frac{1}{2} xs F_{\theta_J}''(b)\).
\end{enumerate}
\end{corollary}

These results motivate the approximation
\beq{eq.approxEFSA}
\expect{F_{\theta_J}(B^J) S^J A^K} \approx 
\bigl( \E{A}[K] \E{S}[J] + \C{AS}[KJ] \bigr) G_J\Bigl( \E{B}[J] + \frac{\C{AB}[KJ]}{\E{A}[K]} + \frac{\C{SB}[JJ]}{\E{S}[J]}, \C{BB}[JJ] \Bigr)
\eeq
as well as a similar one where \(A\) is replaced with \(R\).

\subsection{The dynamical system}

The last approximation together with \cref{eq.approxEFS} yield a closed dynamical system:
\begin{subequations}
\label{eq.DS}
\begin{align}
\label{eq.DS.A}
\dE{A}[J]
    & = - \beta_J \E{A}[J] + \alpha_J \E{S}[J] G_J\Bigl( \E{B}[J] + \frac{\C{SB}[JJ]}{\E{S}[J]}, \C{BB}[JJ] \Bigr), \\
\dE{R}[J]
    & = - \gamma_J \E{R}[J] + \beta_J \E{A}[J], \\[2pt]
\begin{split}
\dC{AA}[JK]
    & = - (\beta_J + \beta_K) \C{AA}[JK] + \alpha_K H_K(\E{A}[J], \E{S}[K], \E{B}[K], \C{AS}[JK], \C{AB}[JK], \C{SB}[KK], \C{BB}[KK]) \\
        &\hspace*{44mm} + \alpha_J H_J(\E{A}[K], \E{S}[J], \E{B}[J], \C{AS}[KJ], \C{AB}[KJ], \C{SB}[JJ], \C{BB}[JJ]),
\end{split} \\[2pt]
\dC{RR}[JK]
    & = - (\gamma_J + \gamma_K) \C{RR}[JK] + \beta_K \C{AR}[KJ] + \beta_J \C{AR}[JK], \\[2pt]
\dC{AR}[JK]
    & = - (\beta_J + \gamma_K) \C{AR}[JK] + \beta_K \C{AA}[JK] + \alpha_J H_J(\E{R}[K], \E{S}[J], \E{B}[J], \C{RS}[KJ], \C{RB}[KJ], \C{SB}[JJ], \C{BB}[JJ]),
\end{align}
\end{subequations}
where 
\[
H_J(x, s, b, c_1, c_2, c_3, v)
    \defeq (xs + c_1) G_J\Bigl( b + \frac{c_2}{x} + \frac{c_3}{s}, v \Bigr) - xs G_J\Bigl( b + \frac{c_3}{s}, v \Bigr).
\]
For a network of \(n\) populations, there are \(2n\) equations for expectations of active and refractory fractions, \(n(n+1)\) equations for covariances between two active or two refractory fractions (including variances), and \(n^2\) equations for covariances between an active and a refractory fraction. Hence, the whole system has \(n(2n+3)\) dimensions.

This dynamical system can be interpreted from the definition of the function \(G_J\). In particular, \cref{eq.DS.A} has the same form as the corresponding equation in the mean-field system, but the activation function is stretched horizontally by the variance of the input and shifted by the covariance between the input and the sensitive fraction. The derivatives of covariances have similar interpretations.

Corollaries to Theorem~\ref{thm.GJ.main} guarantee that the functions \(G_J\) and \(H_J\) are always \(\mathscr{C}^1\) on the relevant domains. Therefore, the vector field corresponding to \cref{eq.DS} is \(\mathscr{C}^1\) on \((0,1)^{2n} \times \mathbb{R}^{n(2n+1)}\). The Picard--Lindelöf theorem then shows the existence and uniqueness of solutions to the differential equation for any initial condition in \((0,1)^{2n} \times \mathbb{R}^{n(2n+1)}\). However, not every initial condition in this domain can be interpreted in terms of the underlying Markov chain, since the boundedness of the fractions of populations implies various bounds on the covariances, for example through the Bhatia--Davis inequality \parencite{bhatia_better_2000} or the Cauchy--Schwarz inequality. Given the various conditions that relate all covariances together, it is not obvious to usefully characterize the domain where solutions certainly make sense from the point of view of the underlying Markov chain, but we do not expect this to be a problem as long as covariances remain small.

It is worth noting the following result about the case where all covariances are zero.

\begin{proposition}
\label{prop.zerocov}
The domain \(\mathscr{D} \defeq (0,1)^{2n} \times \{0\}^{n(2n+1)}\), where \(\C{AA}[JK] = \C{RR}[JK] = \C{AR}[JK] = 0\) for all \(J, K \in \mathscr{P}\), is invariant under the flow of the dynamical system given in \cref{eq.DS}.
\end{proposition}

\begin{proof}
If \(\C{AA}[JK] = \C{RR}[JK] = \C{AR}[JK] = 0\) for all \(J, K \in \mathscr{P}\), then all other covariances are zero as well. It is then clear from the definition of the functions \(H_J\) and from \cref{eq.DS} that \(\dC{AA}[JK] = \dC{RR}[JK] = \dC{AR}[JK] = 0\). Hence, on the domain \(\mathscr{D}\), the vector field corresponding to the dynamical system is parallel to \(\mathscr{D}\), which implies its invariance.
\end{proof}

Proposition~\ref{prop.zerocov} has an important consequence, since by Theorem~\ref{thm.GJ.main}, \(G_J(\E{B}[J], 0) = F_{\theta_J}(\E{B}[J])\). Indeed, it follows that if all covariances are forced to zero, the system reduces to the mean-field system given in \cref{eq.meanfield}. This shows that we can see the mean-field system as a subsystem of that given in \cref{eq.DS}. In particular, any steady state of the mean-field system is also a steady state of the second-order model in which all covariances are zero. Moreover, since the mean-field system can be seen as an extension of Wilson--Cowan's system \parencite{painchaud_beyond_2022}, this means that our second-order system is also an extension of Wilson--Cowan's.

\section{Examples}

We now illustrate how adding covariances can lead to a better prediction of the underlying Markov chain's macroscopic behavior. We first give an example where the second-order model correctly predicts the steady state to which trajectories of the underlying Markov chain converge, whereas the mean-field model does not. Then, we present two examples in which the second-order model converges to a steady state with nonzero covariances that cannot be described by the mean-field model. In both cases, while the mean-field model is able to predict some aspects of the Markov chain's macroscopic dynamics, the second-order model provides more information, leading to a better prediction of the Markov chain's behavior.

\subsection{A word on the methodology}

In all examples, we compare the solutions of the mean-field and second-order models with trajectories of the underlying Markov chain obtained using the Doob--Gillespie algorithm \parencite{gillespie_general_1976, gillespie_exact_1977}. To produce a meaningful comparison between the results of macroscopic and microscopic models, it is crucial to ensure that the microscopic parameters and initial states are consistent with the macroscopic ones. 

For simplicity, we take the microscopic parameters \(\alpha\), \(\beta\) and \(\gamma\) constant over populations and equal to their macroscopic values. In the same way, having chosen a connection coefficient \(c_{JK}\) from a population \(K\) to a population \(J\), we set the weight of the connection from any neuron \(k \in K\) to any \(j \in J\) to be \(W_{jk} = \nicefrac{c_{JK}}{\abs{K}}\). However, we do not assume the thresholds to be constant: we rather assume that thresholds in a population \(J\) follow a logistic distribution with mean \(\theta_J\) and scaling factor \(s_{\theta_J}\), having cumulative distribution function
\[
F_{\theta_J}(x) = \frac{1}{1 + \exp\Bigl( - \dfrac{x - \theta_J}{s_{\theta_J}} \Bigr)}.
\]

Finally, we always specify the distribution of the microscopic initial state from two pieces of information for each population \(J\):
\begin{enumerate}
    \item values \(\E{A}[J](0)\) and \(\E{R}[J](0)\) for the initial expectations of its active and refractory fractions;
    \item a number \(n_J\) of distinct neuronal states assumed to divide its size \(\abs{J}\).
\end{enumerate}
To do this, we define a collection \(\bigl\{Z_J^\ell : J \in \mathscr{P}, \ell \in \{1, \hdots, n_J\} \bigr\}\) of independent random variables taking values \(1\), \(i\) and \(0\) with probabilities \(\E{A}[J](0)\), \(\E{R}[J](0)\) and \(\E{S}[J](0) = 1 - \E{A}[J](0) - \E{R}[J](0)\) respectively. Then, we split each population \(J\) in \(n_J\) subgroups of \(\nicefrac{\abs{J}}{n_J}\) neurons, and set the initial state of all neurons of the \(\ell\)th subgroup as the random variable \(Z_J^\ell\). For example, if we choose \(n_J = \abs{J}\) for each population, this simply amounts to taking the initial states of all neurons as independent random variables which are identically distributed over populations. If instead \(n_J = \nicefrac{\abs{J}}{10}\) for each population, then the network is split in subgroups of 10 neurons whose initial states are perfectly correlated, but the states of these subgroups are still independent and identically distributed over populations.

This completely specifies the distribution of the initial state \(X_0\), thus specifying initial expectations and covariances as well. For example, the expectation of the active fraction of population \(J\) is
\[
\frac{1}{\abs{J}} \sum_{j\in J} \expect[\big]{\Re X_0^j}
    = \frac{1}{\abs{J}} \sum_{\ell=1}^{n_J} \frac{\abs{J}}{n_J}\, \expect[\big]{\Re Z_J^\ell}
    = \frac{1}{n_J} \sum_{\ell=1}^{n_J} \E{A}[J](0)
    = \E{A}[J](0),
\]
as it should be. The same pattern works for the other two states. Then, similarly,
\[
\C{AA}[JJ](0)
    %= \frac{1}{\abs{J}^2} \sum_{j,k\in J} \cov[\big]{\Re X_0^j}{\Re X_0^k}
    = \frac{1}{n_J^2} \sum_{\ell=1}^{n_J} \var[\big]{\Re Z_J^\ell}
    = \frac{1}{n_J} \E{A}[J](0) \bigl( 1 - \E{A}[J](0) \bigr),
\]
and again the same pattern works to compute initial variances for the other two states. Covariances between the fractions of a population in two distinct states can then be computed from initial variances, and since the initial states of neurons from distinct populations are always independent, any initial covariance between fractions of two distinct populations is zero. 

This method of choosing the initial state allows us to ensure that the initial state used with the second-order system is consistent with the initial distribution used in simulations of the underlying stochastic process. The distribution is admittedly somewhat artificial, but it allows to introduce between neuronal states nonzero correlations whose magnitude can be modified through the parameter \(n_J\), while still being convenient for numerical simulations. 

%We stress that we use this method to choose initial states only to ensure that in each example, the macroscopic initial state used in the second-order dynamical system is consistent with the distribution of microscopic initial states used for stochastic simulations, so that we can compare the dynamical system's solution with trajectories of the Markov chain started from the ``same'' initial state. The fact that population sizes intervene in our computation of initial states has no relation with the second-order model itself. 

\subsection{Covariances can improve the model's accuracy}
\label{sec.example1}

\begin{figure*}
\includegraphics[width=\linewidth]{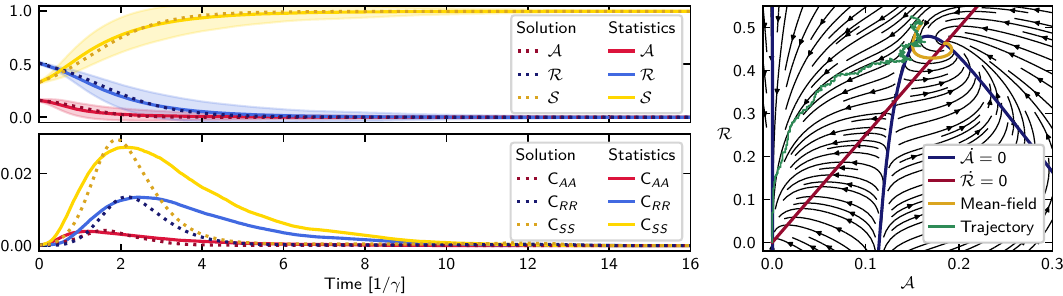}
\caption{Numerical simulations for a network of a single population of 1000 neurons with parameters given in \cref{eq.ex1.params} and initial state given in \cref{eq.ex1.initstate} with \(n = 1000\) distinct initial neuronal states. On the left panels is a comparison between the solution of the second-order dynamical system and statistics computed from 1000 simulated trajectories of the underlying Markov chain. The shaded regions around the curves are associated with statistics, and are bounded above and below by a difference of one standard deviation from the average value. On the right panel is a phase plane of the mean-field dynamical system, on which are plotted the solution of the mean-field system and the macroscopic behavior of a typical trajectory of the Markov chain. See Figure~\ref{fig.ex1.covs} in Appendix~\ref{apx.covs} for the other covariances.}
\label{fig.ex1}
\end{figure*}

The first example is a case where the second-order model allows to predict accurately the macroscopic behavior of the underlying Markov chain, while the mean-field model does not.

Consider a network of \(N\) neurons in a single population, with parameters
\beq{eq.ex1.params}
\begin{aligned}
    \alpha & = 1.4\, [\gamma], &
    \beta & = 2.5\, [\gamma], &
    \gamma & = 1\, [\gamma], \\
    \theta & = 0.75, &
    s_\theta & = 0.1, &
    Q & = 0, &
    c & = 5.5,
\end{aligned}
\eeq
where we dropped the subscripts that would refer to the unique population. Here, we measure characteristic rates in units of \(\gamma\), which is equivalent to measuring time in units of \(\nicefrac{1}{\gamma}\) since every term of the dynamical system is proportional to one of the rates. We fix initial expected values
\beq{eq.ex1.initstate}
(\E{A}, \E{R})(0) = (0.16, 0.51),
\eeq
which determines the distribution of the microscopic initial state (and thus the macroscopic initial covariances) as a function of \(N\). Here, we use a network of \(N = 1000\) neurons with as many distinct initial neuronal states.

Integrating numerically the second-order system with parameters from \cref{eq.ex1.params} and initial state given by \cref{eq.ex1.initstate} yields the solution shown on the left panels of Figure~\ref{fig.ex1}. To avoid overloading the figure, we chose to display only the components of the solution associated with expectations and variances. Since there is only one population, these are in fact enough to determine all components of the solution from the identity \(A + R + S \equiv 1\), by bilinearity and symmetry of the covariance. To provide the complete picture if needed, the other covariances are shown on Figure~\ref{fig.ex1.covs}, in Appendix~\ref{apx.covs}. Then, integrating the mean-field system with the same parameters and initial state yields the solution shown on the phase plane on the right panel of Figure~\ref{fig.ex1}. The two models disagree: the mean-field model predicts that the system converges to a fixed point with about \(20\%\) of active neurons, while the second-order model predicts that all activity will stop.

It is interesting to compare these solutions to statistics computed from trajectories of the underlying Markov chain. Choosing randomly microscopic initial states corresponding to the macroscopic initial state given by \cref{eq.ex1.initstate}, we computed the averages and sample covariances of the three fractions of the network in each of the three states over 1000 simulated trajectories. The results are shown on the left panels of Figure~\ref{fig.ex1}. In this case, the second-order model predicts more accurately the macroscopic behavior of the underlying Markov chain than the mean-field model.

\subsection{The second-order model can describe the averaging of multiple steady states}
\label{sec.example2}

\begin{figure*}
\includegraphics[width=\linewidth]{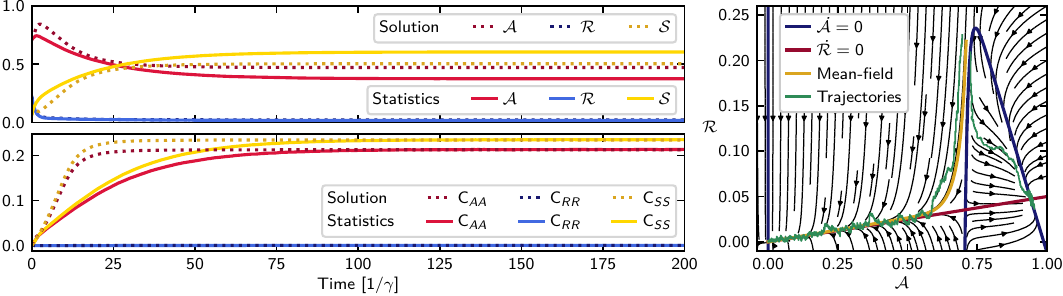}
\caption{Numerical simulations for a network consisting of a single population of 1000 neurons with parameters given in \cref{eq.ex2.params} and initial state given in \cref{eq.ex2.initstate} with \(n = 100\) distinct initial neuronal states. On the left panels is a comparison between the solution of the second-order dynamical system and statistics computed from 1000 simulated trajectories of the underlying Markov chain. On the right panel is a phase plane of the mean-field dynamical system, on which are plotted the solution of the mean-field system and the macroscopic behaviors of two trajectories of the Markov chain. These trajectories are typical representatives of the sets of trajectories that converge to each fixed point. See Figure~\ref{fig.ex2.covs} in Appendix~\ref{apx.covs} for the other covariances.}
\label{fig.ex2}
\end{figure*}

This example illustrates an aspect of the macroscopic dynamics of the Markov chain that can be reproduced by our second-order model and not by the mean-field model.

Consider a network of \(N\) neurons in a single population, with parameters
\beq{eq.ex2.params}
\begin{aligned}
    \alpha & = 4.2\, [\gamma], &
    \beta & = 0.05\, [\gamma], &
    \gamma & = 1\, [\gamma], \\
    \theta & = 12.7, &
    s_\theta & = 0.2, &
    Q & = 0, &
    c & = 17.
\end{aligned}
\eeq
The phase plane of the mean-field system on the right panel of Figure~\ref{fig.ex2} shows that this system has two stable fixed points, one inactive at the origin, and the other very active, with about 95\% of active neurons. Now, fix the initial expectations
\beq{eq.ex2.initstate}
(\E{A},\E{R})(0) = (0.71, 0.221).
\eeq
From this state, integrating the mean-field system yields the solution shown on the phase plane, which converges to the inactive fixed point.

The initial condition given by \cref{eq.ex2.initstate} is very close to the separatrix between the basins of attraction of the two fixed points. As one can expect, it is verified by stochastic simulations of the Markov chain that when starting from a microscopic initial state drawn from a distribution consistent with \cref{eq.ex2.initstate}, a positive variance in the initial state can make trajectories converge to either of the two stable fixed points. Two examples of trajectories are shown on the phase plane on Figure~\ref{fig.ex2}. The second-order model can predict this behavior.

Integrating the second-order system with initial state given by \cref{eq.ex2.initstate} for a network of \(N = 1000\) neurons with \(n = 100\) distinct initial neuronal states yields the solution shown on the left panels of Figure~\ref{fig.ex2}, where as in the first example only variances are displayed, while other covariances are shown on Figure~\ref{fig.ex2.covs} in Appendix~\ref{apx.covs}. It can be seen that the state to which the solution converges is the average of the two stable fixed points of the mean-field system. Moreover, the variances of the active and sensitive fractions of the network are close to \(\nicefrac{1}{4}\), which is seen from the Bhatia--Davis inequality \parencite{bhatia_better_2000} to be the upper bound on the variance of random variables supported on \([0,1]\). To interpret what it means for the variances of \(A_t\) and \(S_t\) to have such high values, recall that the maximal variance of \(\nicefrac{1}{4}\) is attained by a random variable that is \(0\) or \(1\) with probabilities \(\nicefrac{1}{2}\). Thus, the results shown on the left panels of Figure~\ref{fig.ex2} can represent that when \(t\) is large enough, \(A_t\) and \(S_t\) are close to binary random variables with values at the two fixed points of the mean-field system.

To compare this solution with the macroscopic behavior of the underlying Markov chain, we generated 1000 trajectories of the stochastic process with a network of \(N = 1000\) neurons and \(n = 100\) distinct initial neuronal states. The relevant statistics from these trajectories are shown on the left panels of Figure~\ref{fig.ex2}. The results of the stochastic simulations agree well with the prediction of the second-order model. On the other hand, the mean-field model is only a good approximation of the trajectories that converge to the inactive fixed point---it fails to capture the other possible outcome.

\subsection{The second-order model can describe the averaging of oscillations}
\label{sec.example3}

This last example illustrates again the ability of the second-order model to carry information about the distribution of the underlying Markov chain that is inaccessible to the mean-field model.

Consider a network split into two populations \(E\) and \(I\), the former being excitatory with parameters
\begin{subequations}
\label{eq.ex3.params}
\begin{align}
\SwapAboveDisplaySkip
\label{eq.ex3.paramsE}
    \alpha_E & = 0.75\, [\gamma_E], &
    \beta_E & = 0.15\, [\gamma_E], &
    \gamma_E & = 1\, [\gamma_E], \\
    \theta_E & = 0.7, &
    s_{\theta_E} & = 0.2, &
    Q_E & = 0,
\intertext{and the latter being inhibitory with parameters}
\label{eq.ex3.paramsI}
    \alpha_I & = 0.4\, [\gamma_E], &
    \beta_I & = 0.12\, [\gamma_E], &
    \gamma_I & = 0.5\, [\gamma_E], \\
    \theta_I & = 1.8, &
    s_{\theta_I} & = 0.2, &
    Q_I & = 0.
\end{align}
The connections between these populations are given by
\beq{eq.ex3.paramsc}
c = \begin{pmatrix}
    c_{EE} & c_{EI} \\
    c_{IE} & c_{II}
\end{pmatrix} = \begin{pmatrix}
    11 & -12 \\
    12 & -9
\end{pmatrix}.
\eeq
\end{subequations}
Now, fix the initial state from the expectations
\beq{eq.ex3.initstate}
(\E{A}[E], \E{A}[I], \E{R}[E], \E{R}[I])(0) = (0.25, 0.3, 0.2, 0.25).
\eeq
The solution of the mean-field system converges to a limit cycle, as shown on the top left panel of Figure~\ref{fig.ex3}. This solution approximates reasonably well the behavior of the underlying Markov chain, as seen from the example of trajectory shown on the top right panel of Figure~\ref{fig.ex3}.

\begin{figure*}[t!]
\includegraphics[width=\linewidth]{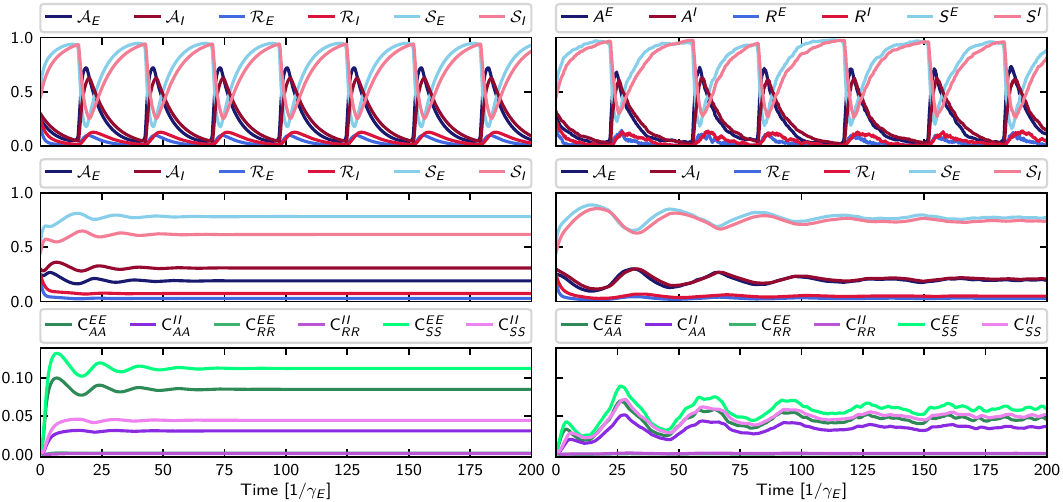}
\caption{Numerical simulations for a network of two populations with parameters given in \cref{eq.ex3.params} and initial state given in \cref{eq.ex3.initstate}, with 500 neurons and 50 distinct initial neuronal states in each population. On the top left panel is the solution of the mean-field dynamical system. On the top right panel is the macroscopic behavior of a typical trajectory of the underlying Markov chain obtained from a numerical simulation. On the middle and bottom panels is a comparison between the solution of the second-order system (left panels) and statistics computed from 1000 simulated trajectories of the Markov chain (right panels). See Figure~\ref{fig.ex3.covs} in Appendix~\ref{apx.covs} for the other covariances.}
\label{fig.ex3}
\end{figure*}

Yet, one can wonder what happens when statistics are computed from the macroscopic behavior of multiple trajectories. As the trajectories are random, the oscillation periods are expected to vary among trajectories, and even among oscillations in each trajectory. Thus, one should expect the oscillations to decay over time when averaging many trajectories. This is the behavior that our second-order model reproduces.

Integrating the second-order system with initial state given by \cref{eq.ex3.initstate} for populations of 500 neurons each with an initial distribution of \(n_E = n_I = 50\) distinct neuronal states yields the solution shown on the middle and bottom left panels of Figure~\ref{fig.ex3}. As in the other examples, in order to avoid overloading the figure we showed only expectations and variances on the figure as these variables are the easier to interpret, but now they are not enough to determine covariances between fractions of two distinct populations in given states. The other covariances are shown on Figure~\ref{fig.ex3.covs}, in Appendix~\ref{apx.covs}. The solution converges to a fixed point where variances are nonzero, and where the value of each expectation is about the average, over a cycle, of the corresponding expectation predicted by the mean-field system. To compare this solution with the macroscopic behavior of the underlying Markov chain, we generated 1000 trajectories in a network with 500 neurons in each population. The relevant statistics from these trajectories are shown on the middle and bottom right panels of Figure~\ref{fig.ex3}.% and illustrate that, at least qualitatively, the second-order model succeeds in predicting the macroscopic behavior of the Markov chain.

Our second-order model succeeds in predicting the averaging of the trajectories and the convergence of the system to a fixed point, although it fails to capture correctly the system's behavior in the transient period. This could be due to the fact that the network's size has a sizable effect on the transient behavior. Indeed, the rate at which the average of trajectories converges to the average of the oscillations depends on the relative size of the noise with respect to the oscillation amplitude, which depends on the network's size: the less neurons there are in the network, the faster the oscillations will dampen. However, regardless of the network's size, the stochastic trajectories should always become distributed all over the cycle after some time, and this is what our second-order model captures, at least qualitatively. It is possible that in order to increase the precision of the prediction, one would have to include higher-order moments in the model.

\section{Discussion}

We demonstrated the importance of accounting for covariances between states of neural populations when modeling the activity of large stochastic neural networks. Our observations were based on the low-dimensional dynamical system in \cref{eq.DS}, which we derived mostly by developing a novel method for estimating the expected value of random variables as in \cref{eq.approxEFS,eq.approxEFSA}, and solving the corresponding moment closure problem. We achieved this while explicitly describing the proportion of refractory neurons in each population. Our numerical experiments showed that in some cases, the mean-field solutions greatly differ from both the average of realizations of the stochastic model and the solutions of our second-order model, and that our model succeeds in improving the prediction of the network's activity as compared to the mean-field model.

There is a first class of examples (see Figures~\ref{fig.ex1} and \ref{fig.ex2}) in which the inclusion of covariances leads to predictions that match much more closely the behavior of the high-dimensional stochastic system. This can occur even if after a brief transient period, all covariances go to zero. The presence of two attractors in the mean-field system explains well this phenomenon. Indeed, if the initial conditions lie close to the separatrix between the basins of attraction, the qualitative features of the solution can be highly sensitive to small perturbations of the early system behavior. Failing to account for covariances is akin to not using all the available information about the system's initial state, which may lead to predicting convergence to the wrong attractor. A lot of research has been done recently for systems with only one attracting point, in which case injecting information about covariances increases the order of convergence \parencite{gast_size_2019}. In our case, covariances have more dramatic effects as they may make the system jump basin of attraction.

There is yet another situation in which the covariances may dramatically affect the solutions, as some covariances can remain positive throughout the simulations. As shown on Figure~\ref{fig.ex3}, the solution of the mean-field system can display sustained oscillations while the solution of the second-order model converges to a fixed point. This can be explained by the presence of positive covariances, which impacts the activation functions of neuron populations by making them more linear. Since the presence of multiple fixed points or oscillations in the mean-field model depends on the nonlinearity of the activation function, changing its shape is likely to change the topology of the solutions. 

Despite the improvement achieved by our second-order model in predicting the macroscopic behavior of the stochastic process as compared to the mean-field model, it appears from the results shown on Figures~\ref{fig.ex2} and \ref{fig.ex3} that our second-order model can still not predict the macroscopic behavior of the underlying stochastic process with complete accuracy. Given the nonlinearity of our model and the presence of multiple attractors, obtaining rigorous theorems on rates of convergence or on the quality of our dimension reduction would be difficult. Indeed, the accuracy of the reduction of a high-dimensional stochastic system to a low-dimensional dynamical system is not easy to determine analytically. In the case where the system has a single attracting point, this problem has recently been investigated by \textcite{gast_refined_2017, gast_size_2019, gast_refinements_2020}, but we are not aware of more general results. In order to obtain a more accurate model, it is possible that one could build a different second-order solution to the moment closure problem posed by \cref{eq.DSopen}, but it is also possible that one would have to find a higher-order solution involving at least third-order moments.

We argue for the necessity of taking into account the inherently stochastic nature of neuron activity. This stochastic nature does not come from the way that mathematical modeling of neurons is performed, but from omnipresent biological mechanisms such as vesicle release or random channel state transition. Our results show that using a naive approach to handle this stochasticity, such as a mean-field approximation, can lead to false conclusions, and that a second-order solution to the moment closure problem posed by \cref{eq.DSopen} is more appropriate in several situations. Our work thus emphasizes the importance of pursuing the study of higher moments in models of neuronal activity. In particular, the observation that the solutions of our second-order model correspond to the averages of many trajectories of the stochastic model invites further theoretical investigation. We believe that the mathematical tools developed here will help deriving more effective low-dimensional approximations of stochastic processes on networks, contributing to a better understanding of large-scale dynamical phenomena in complex systems.

\subsubsection*{Code Availability}

All numerical results presented in this paper were obtained with the PopNet package \parencite{painchaud_popnet_2022}, written in the Python programming language and available on GitHub.

\subsubsection*{Acknowledgments}

This work was supported by the Natural Sciences and Engineering Research Council of Canada (N.D, P.D, V.P), the Fonds de recherche du Qu\'ebec -- Nature and technologies (N.D., V.P.), and the Sentinel North program of Universit\'e Laval (N.D., P.D), funded by the Canada First Research Excellence Fund.

\appendix
\setcounter{theorem}{0}
\setcounter{corollary}{0}
\renewcommand{\theHtheorem}{A\thetheorem}
\renewcommand{\theHcorollary}{A\thecorollary}

\section{Derivation of the nonautonomous microscopic and macroscopic evolution equations}
\label{apx.DS}

In this section, we derive the nonautonomous evolution equations given in the main text in \cref{eq.microsys,eq.DSopen,eq.DScovs}. In order to derive \cref{eq.DScovs}, we obtain as an intermediate step yet another system that can be seen as a second-order version of that given in \cref{eq.microsys}. We derive each of these systems in the first four subsections, and then we conclude the section by some remarks on the limits of the approximations that we use to obtain the macroscopic systems.

\subsection{The nonautonomous microscopic first-order evolution equation}
\label{apx.DSmicro1}

First, we derive the expressions of the derivatives of the probabilities
\begin{alignat*}{2}
p_j(t) & \defeq \bprob{X_t^j = 1} && = \expect[\big]{\Re X_t^j}, \\
r_j(t) & \defeq \bprob{X_t^j = i} && = \expect[\big]{\Im X_t^j}, \\
q_j(t) & \defeq \bprob{X_t^j = 0} && = \expect[\big]{1 - \abs{X_t^j}}.
\end{alignat*}
where, as in the main text, the continuous-time Markov chain \(\{X_t\}_{t\geq 0}\) takes its values in \(E \defeq \{0,1,i\}^N\), \(N\) being the number of neurons in the network. An entry \(X_t^j\) is then interpreted as the state of neuron \(j\) at time \(t\).

To obtain the expression of the derivative of \(p_j(t)\), let \(\Delta t > 0\). Conditioning on the state of the network at time \(t\),
\begin{align*}
p_j(t + \Delta t)
    & = \sum_{x\in E} \bprob{X_{t+\Delta t}^j = 1 \given X_t = x} \bprob{X_t = x} \\
    & = \sum_{x\in E} \bigl( \Re x_j + \Im x_j + 1 - \abs{x_j} \bigr) \bprob{X_{t+\Delta t}^j = 1 \given X_t = x} \bprob{X_t = x}.
\end{align*}
Because \(x_j\) has a value in \(\{0,1,i\}\), exactly one of \(\Re x_j\), \(\Im x_j\) or \(1 - \abs{x_j}\) is 1 while the other two are 0. Hence, we can now split the sum in three parts: one for each possible state of \(j\) at time \(t\). Each conditional probability \(\bprob{X_{t+\Delta t}^j = 1 \given X_t = x}\) can then be expressed from transition rates according to the relations given in \cref{eq.transitionprobabilities}:
\begin{align*}
\bprob{X_{t+\Delta t}^j = 1 \given X_t = x, X_t^j = 0}
    & = a_j(\eta,x) \Delta t + o(\Delta t), \\
\bprob{X_{t+\Delta t}^j = 1 \given X_t = x, X_t^j = 1}
    & = 1 - \beta_j(\eta) \Delta t + o(\Delta t), \\
\bprob{X_{t+\Delta t}^j = 1 \given X_t = x, X_t^j = i}
    & = o(\Delta t).
\end{align*}
Using these relations, we obtain that as \(\Delta t \downarrow 0\),
\begin{align*}
p_j(t + \Delta t)
    & = \sum_{x\in E} \Re x_j \bigl( 1 - \beta_j(\eta) \Delta t \bigr) \bprob{X_t = x} 
        + \sum_{x\in E} \bigl( 1 - \abs{x_j} \bigr) a_j(\eta, x) \Delta t \bprob{X_t = x} + o(\Delta t).
\intertext{These sums are now expectations of functions of \(X_t\), and we find that}
p_j(t + \Delta t)
    & = \bigl( 1 - \beta_j(\eta) \Delta t \bigr) p_j(t) + \Delta t \expect[\big]{\bigl( 1 - \abs{X_t^j} \bigr) a_j(\eta, X_t)} + o(\Delta t).
\end{align*}
as \(\Delta t \downarrow 0\). Rearranging and dividing through by \(\Delta t\) yields
\[
\dot{p}_j(t) = - \beta_j(\eta) p_j(t) + \expect[\big]{\bigl( 1 - \abs{X_t^j} \bigr) a_j(\eta, X_t)}.
\]
The derivatives of \(r_j\) and \(q_j\) are obtained in the same way. In these two cases, the transition probabilities \(\bprob{X_{t+\Delta t}^j = i \given X_t = x}\) and \(\bprob{X_{t+\Delta t}^j = 0 \given X_t = x}\) are expressed from other transition rates, and we obtain
\begin{align*}
\SwapAboveDisplaySkip
\dot{r}_j(t) & = - \gamma_j(\eta) r_j(t) + \beta_j(\eta) p_j(t)
\shortintertext{and}
\dot{q}_j(t) & = - \expect[\big]{\bigl( 1 - \abs{X_t^j} \bigr) a_j(\eta, X_t)} + \gamma_j(\eta) r_j(t).
\end{align*}

\subsection{The nonautonomous macroscopic first-order evolution equation}

Now, we use the expressions of the derivatives of \(p_j\), \(q_j\) and \(r_j\) obtained above to find the derivatives of the expected fractions of populations
\[
\E{A}[J](t) \defeq \expect{A_t^J}, \qquad
\E{R}[J](t) \defeq \expect{R_t^J} \qquad\text{and}\qquad
\E{S}[J](t) \defeq \expect{S_t^J},
\]
where \(A_t^J\), \(R_t^J\) and \(S_t^J\) are respectively the fractions of active, refractory and sensitive neurons in population \(J\) at time \(t\). 

To obtain an expression for the derivative of \(\E{A}[J]\), we use the idea described in the main text: we average the derivatives \(\dot{p}_j\) over \(J\) using the linearity of expectations and derivatives. From the expression of \(\dot{p}_j\), this first results in
\[
\dE{A}[J](t)
    = \expect[\bigg]{- \frac{1}{\abs{J}} \sum_{j\in J} \beta_j(\eta) \Re X_t^j + \frac{1}{\abs{J}} \sum_{j\in J} a_j(\eta, X_t) \bigl( 1 - \abs{X_t^j} \bigr)}.
\]
Introducing for \(\xi \in \{0,1,i\}\) the subpopulations \(J_t^\xi \defeq \{j \in J : X_t^j = \xi\}\), the last equality can be written as
\beq{eq.dAJintermediate}
\dE{A}[J](t)
    = \expect[\bigg]{- \frac{A_t^J}{\abs{J_t^1}} \sum_{j \in J_t^1} \beta_j(\eta) + \frac{S_t^J}{\abs{J_t^0}} \sum_{j\in J_t^0} a_j(\eta, X_t)},
\eeq
where we used the identities \(\abs{J_t^1} = \abs{J} A_t^J\) and \(\abs{J_t^0} = \abs{J} S_t^J\). Now, assuming that the number of neurons in \(J\) is large, we expect the number of neurons in \(J\) that are active at time \(t\) to be large as well. Since the random variables \(\beta_j\) for \(j \in J_t^1\) are independent and identically distributed, the law of large numbers motivates the approximation
\beq{eq.approxbeta}
\frac{1}{\abs{J_t^1}} \sum_{j\in J_t^1} \beta_j(\eta) \approx \expectmu{\beta_j} \eqdef \beta_J,
\eeq
where \(\mathbb{E}_\mu\) denotes the expectation on \((H,\mathscr{H},\mu)\). The other term can be handled in a similar way. First, we approximate the input
\[
\sum_{k=1}^N W_{jk}(\eta) \Re X_t^k + Q_J
    = \sum_{K\in\mathscr{P}} \sum_{k\in K_t^1} W_{jk}(\eta) + Q_J
    \approx \sum_{K\in\mathscr{P}} \abs{K} A_t^K \expectmu{W_{jk}} + Q_J.
\]
To simplify notation, we define the input in population \(J\) at time \(t\)
\[
B_t^J \defeq \sum_{K\in\mathscr{P}} c_{JK} A_t^K + Q_J
\qquad\text{with}\qquad
c_{JK} \defeq \abs{K}\expectmu{W_{jk}}
\]
for \(j \in J\) and \(k \in K\). This leads to approximate, for \(j \in J\),
\[
a_j(\eta, X_t) \approx \alpha_j(\eta) \charf{\{B_t^J > \theta_j\}}(\eta),
\]
and the law of large numbers now motivates the approximation
\beq{eq.approxalpha}
\frac{1}{\abs{J_t^0}} \sum_{j\in J_t^0} a_j(\eta, X_t) \approx \expectmu{\alpha_j\charf{\{B_t^J > \theta_j\}}} = \alpha_J F_{\theta_J}(B_t^J),
\eeq
where \(\alpha_J \defeq \expectmu{\alpha_j}\) and  \(F_{\theta_J}\) denotes the cumulative distribution function of \(\theta_j\) for \(j \in J\). 

Using the approximations from \cref{eq.approxbeta,eq.approxalpha} in the expression of \(\dE{A}[J]\) given in \cref{eq.dAJintermediate} yields
\[
\dE{A}[J](t) \approx - \beta_J \E{A}[J](t) + \alpha_J \expect{S_t^J F_{\theta_J}(B_t^J)}.
\]
The same method can be applied to find approximate expressions for the derivatives of \(\E{R}[J]\) and \(\E{S}[J]\). This leads us to model the macroscopic dynamics of the network by the differential equations \eqref{eq.DSopen}.

\subsection{The nonautonomous microscopic second-order evolution equation}

As a first step towards finding the derivatives of covariances given in \cref{eq.DScovs}, we now find expressions for the derivatives of the probabilities
\begin{align*}
p_{jk}(t)
    & \defeq \bprob{X_t^j = X_t^k = 1} 
    = \expect{\Re X_t^j \Re X_t^k}, \\
r_{jk}(t)
    & \defeq \bprob{X_t^j = X_t^k = i} 
    = \expect{\Im X_t^j \Im X_t^k}, \\
\rho_{jk}(t)
    & \defeq \bprob{X_t^j=1, X_t^k=i}
    = \expect{\Re X_t^j \Im X_t^k}.
\end{align*}
The idea is exactly the same as in the case of the probabilities \(p_j\), \(r_j\) and \(q_j\). We will give the details of the calculations only for the case of \(\rho_{jk}\), as it illustrates all of the relevant ideas. 

Let \(\Delta t > 0\). Then
\[
\rho_{jk}(t + \Delta t) = \begin{multlined}[t]
    \sum_{x\in E} \bigl( \Re x_j + \Im x_j + 1 - \abs{x_j} \bigr) \bigl( \Re x_k + \Im x_k + 1 - \abs{x_k} \bigr) \\
    \times \bprob{X_{t+\Delta t}^j = 1, X_{t+\Delta t}^k = i \given X_t = x} \bprob{X_t = x}.
\end{multlined}
\]
The sum now gets split into nine parts, one for each of the possible states for the pair \((j,k)\). Again, the conditional probability \(\bprob{X_{t+\Delta t}^j = 1, X_{t+\Delta t}^k = i \given X_t = x}\) can be expressed in terms of transition rates. These expressions are all derived from the general statement that for any \(x,y \in E\), as \(\Delta t \downarrow 0\),
\[
\bprob{X_{t+\Delta t} = y \given X_t = x} = \delta_{xy} + m^\eta(x,y) \Delta t + o(\Delta t),
\]
where \(m^\eta(x,y)\) is the \((x,y)\) entry of the generator \(M^\eta\), as in the main text. It follows from this statement that as \(\Delta t \downarrow 0\),
\begin{align*}
\bprob{X_{t+\Delta t}^j = 1, X_{t+\Delta t}^k = i \given X_t = x, X_t^j = 1, X_t^k = i}
    & = 1 - \bigl( \beta_j(\eta) + \gamma_k(\eta) \bigr) \Delta t + o(\Delta t), \\
\bprob{X_{t+\Delta t}^j = 1, X_{t+\Delta t}^k = i \given X_t = x, X_t^j = 0, X_t^k = i}
    & = a_j(\eta, x) \Delta t + o(\Delta t), \\
\bprob{X_{t+\Delta t}^j = 1, X_{t+\Delta t}^k = i \given X_t = x, X_t^j = 1, X_t^k = 1}
    & = \beta_k(\eta) \Delta t + o(\Delta t),
\end{align*}
while the transition probabilities from states \(x\) where \(j\) and \(k\) have other states are all \(o(\Delta t)\). Using these relations, we see that as \(\Delta t \downarrow 0\),
\begin{align*}
\rho_{jk}(t+\Delta t)
    & = \sum_{x\in E} \Re x_j \Im x_k \bigl( 1 - \bigl( \beta_j(\eta) + \gamma_k(\eta) \bigr) \Delta t \bigr)  \bprob{X_t = x} \\
    & \hspace*{22mm} + \sum_{x\in E} \bigl( 1 - \abs{x_j} \bigr) \Im x_k\, a_j(\eta, x) \Delta t \bprob{X_t = x} \\
    & \hspace*{22mm} + \sum_{x\in E} \Re x_j \Re x_k\, \beta_k(\eta) \Delta t \bprob{X_t = x} + o(\Delta t) \\
    & = \bigl( 1 - \bigl( \beta_j(\eta) + \gamma_k(\eta) \bigr) \Delta t \bigr) \rho_{jk}(t) \\
    & \hspace*{22mm} + \expect[\big]{\bigl( 1 - \abs{X_t^j} \bigr) \Im X_t^k a_j(\eta, X_t)} \Delta t + \beta_k(\eta) p_{jk}(t) \Delta t + o(\Delta t).
\end{align*}
Rearranging and dividing through by \(\Delta t\) finally yields
\begin{subequations}
\label{eq.microsys2}
\begin{align}
\dot{\rho}_{jk}(t)
    & = - \bigl( \beta_j(\eta) + \gamma_k(\eta) \bigr) \rho_{jk}(t) + \expect[\big]{\bigl( 1 - \abs{X_t^j} \bigr) \Im X_t^k a_j(\eta, X_t)} + \beta_k(\eta) p_{jk}(t).
\intertext{The same method allows to find}
\begin{split}
\dot{p}_{jk}(t)
    & = - \bigl( \beta_j(\eta) + \beta_k(\eta) \bigr) p_{jk}(t) + \expect[\big]{\bigl( 1 - \abs{X_t^j} \bigr) \Re X_t^k a_j(\eta, X_t)} \\
    &\hspace*{55mm} + \expect[\big]{\Re X_t^j \bigl( 1 - \abs{X_t^k} \bigr) a_k(\eta, X_t)}, 
\end{split} \\
\dot{r}_{jk}(t)
    & = - \bigl( \gamma_j(\eta) + \gamma_k(\eta) \bigr) r_{jk}(t) + \beta_j(\eta) \rho_{jk}(t) + \beta_k(\eta) \rho_{kj}(t).
\end{align}
\end{subequations}
We could find similar equations to describe the evolution of similar probabilities with \(j\) or \(k\) being sensitive, but these can always be expressed as functions of \(p_{jk}\), \(r_{jk}\), \(\rho_{jk}\) and \(\rho_{kj}\) since each entry \(X_t^j\) is always in \(\{0,1,i\}\).

\subsection{The nonautonomous macroscopic second-order evolution equation}

We now find the derivatives of covariances between active and refractory fractions of populations of the network. We first find the derivatives of the expectations
\[
\EE{AA}[JK](t) \defeq \expect{A_t^J A_t^K}, \qquad
\EE{RR}[JK](t) \defeq \expect{R_t^J R_t^K} \qquad\text{and}\qquad
\EE{AR}[JK](t) \defeq \expect{A_t^J R_t^K}
\]
from the differential equations \eqref{eq.microsys2}, using the same strategy as for first moments. As in the microscopic case we only detail the case of \(\EE{AR}[JK]\), as the other two are similar.

By linearity of expectations and derivatives,
\begin{align*}
\dEE{AR}[JK](t)
    & = \mathbb{E}^\eta\biggl[ \frac{1}{\abs{J} \abs{K}} \sum_{j \in J, k \in K} \Bigl(
        - \bigl( \beta_j(\eta) + \gamma_k(\eta) \bigr) \Re X_t^j \Im X_t^k \\
        &\hspace*{39mm} + \bigl( 1 - \abs{X_t^j} \bigr) \Im X_t^k a_j(\eta, X_t) + \beta_k(\eta) \Re X_t^j \Re X_t^k
    \Bigr) \biggr] \\
    & = \mathbb{E}^\eta\biggl[ 
        - \frac{A_t^J R_t^K}{\abs{J_t^1}} \sum_{j\in J_t^1} \beta_j(\eta) - \frac{A_t^J R_t^K}{\abs{K_t^i}} \sum_{k\in K_t^i} \gamma_k(\eta) \\
        &\hspace*{39mm} + \frac{S_t^J R_t^K}{\abs{J_t^0}} \sum_{j\in J_t^0} a_j(\eta, X_t) + \frac{A_t^J A_t^K}{\abs{K_t^1}} \sum_{k\in K_t^1} \beta_k(\eta)
    \biggr].
\end{align*}
In the same way as in the case of first-order moments, the law of large numbers motivates the approximations
\[
\frac{1}{\abs{J_t^1}} \sum_{j\in J_t^1} \beta_j(\eta) \approx \beta_J, \qquad
\frac{1}{\abs{K_t^i}} \sum_{k\in K_t^i} \gamma_k(\eta) \approx \gamma_K
\]
and
\[
\frac{1}{\abs{J_t^0}} \sum_{j\in J_t^0} a_j(\eta, X_t) \approx \alpha_J F_{\theta_J}(B_t^J),
\]
and we find the approximate equation
\begin{subequations}
\label{eq.DSopen2}
\begin{align}
\dEE{AR}[JK](t)
    & = - (\beta_J + \gamma_K) \EE{AR}[JK](t) + \alpha_J \expect{S_t^J R_t^K F_{\theta_J}(B_t^J)} + \beta_K \EE{AA}[JK](t).
\intertext{The same method leads to}
\dEE{AA}[JK](t)
    & = - (\beta_J + \beta_K) \EE{AA}[JK](t) + \alpha_J \expect{S_t^J A_t^K F_{\theta_J}(B_t^J)} + \alpha_K \expect{A_t^J S_t^K F_{\theta_K}(B_t^K)}, \\
\dEE{RR}[JK](t)
    & = - (\gamma_J + \gamma_K) \EE{RR}[JK](t) + \beta_J \EE{AR}[JK](t) + \beta_K \EE{AR}[KJ](t).
\end{align}
\end{subequations}

By definition of the covariance, \(\C{AR}[JK] = \EE{AR}[JK] - \E{A}[J]\E{R}[K]\), and similar relations hold for other combinations of fractions of populations. Therefore, the expressions of the derivatives of the covariances \(\C{AA}[JK]\), \(\C{RR}[JK]\) and \(\C{AR}[JK]\) given in \cref{eq.DScovs} follow from \cref{eq.DSopen2} along with the derivatives of \(\E{A}[J]\) and \(\E{R}[J]\) given in \cref{eq.DSopen}.

\subsection{Limits of the approximations}

We start by remarking that, even though we motivate the approximations of the averages of transition rates (e.g.\@ in \cref{eq.approxbeta,eq.approxalpha}) by the law of large numbers, it would not be obvious to make the argument fully rigorous by taking a limit where the sizes of all populations grow infinitely large. The reason for this is that the sets that are needed to be large, which are the subpopulations of active, sensitive and refractory neurons in each population, are themselves random on \((\Omega, \mathscr{F}, \mathbb{P}^\eta)\), but the probability measure \(\mathbb{P}^\eta\) is a function on \((H, \mathscr{H}, \mu)\). Therefore, the setting cannot be easily translated to a simple sequence of independent and identically distributed random variables, and we settle for understanding \cref{eq.approxbeta,eq.approxalpha} as approximations. 

We expect the approximations
\[
\frac{1}{\abs{J_t^1}} \sum_{j\in J_t^1} \beta_j(\eta) \approx \beta_J
\qquad\text{and}\qquad
\frac{1}{\abs{J_t^i}} \sum_{j\in J_t^i} \gamma_j(\eta) \approx \gamma_J
\]
to be good approximations for large populations regardless of the specific distributions of \(\beta_j\) and \(\gamma_j\) in \(J\), as long as these remain independent and identically distributed in \(J\). Indeed, we do not expect the specific values of \(\beta_j\) and \(\gamma_j\) to have an impact on the probabilities for \(j\) to be active and refractory at time \(t\), so we expect the empirical distributions of \(\beta_j\) and \(\gamma_j\) over \(J_t^1\) and \(J_t^i\) to approximate reasonably well the actual distributions when \(J\) is large.

However, the case of the corresponding approximations for the activation rates,
\beq{eq.approxactivation}
\frac{1}{\abs{J_t^0}} \sum_{j\in J_t^0} \alpha_j(\eta) \charf{T_j(X_t)}(\eta)
    \approx \frac{1}{\abs{J_t^0}} \sum_{j\in J_t^0} \alpha_j(\eta) \charf{\{B_t^J > \theta_j\}}(\eta)
    \approx \alpha_J F_{\theta_J}(B_t^J),
\eeq
requires more care. Assuming that every neuron of \(J\) is given the same input \(B_t^J\), the last approximation should hold only when the distribution of \(\theta_j\) in \(J\) is not too spread out. Indeed, only the neurons in \(J\) whose thresholds are lower than their input have a nonzero probability to activate, so it is possible for the specific value of the threshold of a neuron to have an effect on the probability that it is sensitive at a given time. This could lead the empirical distribution of the thresholds in the subpopulation \(J_t^0\) to be biased towards higher thresholds. Nevertheless, since the activation probabilities only depend on the sign of the difference between the input and the threshold, we do not expect this effect to be important when the distribution of the thresholds is peaked enough. Similar arguments also imply that the first approximation in \cref{eq.approxactivation} should only hold when the distributions of the weights \(W_{jk}\) are peaked enough.

\section{Proof of Theorem~\ref{thm.GJ} and its corollaries}
\label{apx.proofs}

In this section, we prove Theorem~\ref{thm.GJ} and its corollaries. For convenience, we recall the general setup.

We study here a function \(F_{\theta_J} \colon \mathbb{R} \to [0,1]\), which is the cumulative distribution function of the thresholds in a population \(J\) of the network. We assume that this distribution is unimodal and symmetric with mean \(\theta_J\). Then, we define the function \(g_J\colon \mathbb{R} \times [0,\infty)\) by setting
\beq{eq.defgJ}
g_J(b,v) \defeq \frac{v}{2(\theta_J - b)} \frac{F_{\theta_J}''(b)}{F_{\theta_J}'(b)}
\eeq
for \(b \neq \theta_J\), assuming \(F_{\theta_J}' > 0\). Since the distribution is unimodal and symmetric around \(\theta_J\), it must be that \(F_{\theta_J}''(\theta_J) = 0\), and \(g_J\) can be continuously extended to \(b = \theta_J\) by definition of the derivative of \(F_{\theta_J}''\), as long as it exists. Finally, we define \(G_J\colon \mathbb{R} \times [0,\infty)\) as
\beq{eq.defGJ}
G_J(b,v) \defeq F_{\theta_J}\Bigl( \frac{b + \theta_J g_J(b,v)}{1 + g_J(b,v)} \Bigr).
\eeq

We now recall Theorem~\ref{thm.GJ} and prove it.

\begin{theorem}
\label{thm.GJ}
Suppose that the thresholds in population \(J\) follow a unimodal and symmetric distribution with mean \(\theta_J\) and cumulative distribution function \(F_{\theta_J}\). Let \(g_J\) and \(G_J\) be defined by \cref{eq.defgJ,eq.defGJ} respectively. Suppose that
\begin{enumerate}[label=\enumlabelformat{\arabic*}]
    \item\label{thm.GJ.conditionregularity} \(F_{\theta_J}\) is \(\mathscr{C}^4\) on \(\mathbb{R}\);
    \item\label{thm.GJ.conditionpositivity} \(F_{\theta_J}' > 0\);
    \item\label{thm.GJ.conditionboundedness} \(g_J(\cdot, 1)\) is bounded on \(\mathbb{R}\);
    \item\label{thm.GJ.conditionDgJ} \(\forall b \in \mathbb{R}\), \(g_J(b,1) + (\theta_J - b) \partial_1 g_J(b,1) \geq 0\).
\end{enumerate}
Then \(G_J\) satisfies the following conditions.
\begin{enumerate}
    \item\label{thm.GJ.regularity} \(G_J\) is \(\mathscr{C}^1\) on \(\mathbb{R} \times [0,\infty)\);
    \item\label{thm.GJ.extremevariance} \(G_J(\cdot, 0) = F_{\theta_J}\) and \(\forall b \in \mathbb{R}, G_J(b, v) \to \nicefrac{1}{2}\) as \(v \to \infty\);
    \item\label{thm.GJ.extremeinput} \(\forall v \geq 0, G_J(b,v) \to 0\) as \(b \to -\infty\) and \(G_J(b,v) \to 1\) as \(b \to \infty\);
    \item\label{thm.GJ.increasing} \(\forall v \geq 0\), \(G_J(\cdot,v)\) is increasing;
    \item\label{thm.GJ.secondderivative} \(\partial_2G_J(\cdot,0) = \frac{1}{2} F_{\theta_J}''\).
\end{enumerate}
\end{theorem}

\begin{proof}
We start by proving the second and third properties. First, note that for any \(b \in \mathbb{R}\), \(g_J(b,0) = 0\) so \(G_J(b,0) = F_{\theta_J}(b)\). Then, as \(v \to \infty\), \(g_J(b,v) = v g_J(b,1) \to \infty\), which implies that
\[
\frac{b + \theta_J g_J(b,v)}{1 + g_J(b,v)} \to \theta_J
\qquad\text{so that}\qquad
G_J(b,v) \to F_{\theta_J}(\theta_J) = \frac{1}{2},
\]
proving \ref{thm.GJ.extremevariance}. Similarly, the fact that \(g_J(\cdot,1)\) is bounded implies that \(g_J\) is bounded with respect to its first argument, so
\[
\frac{b + \theta_J g_J(b,v)}{1 + g_J(b,v)} \to \pm\infty
\]
as \(b \to \pm\infty\), and property \ref{thm.GJ.extremeinput} follows by properties of a cumulative distribution function.

To prove the other three properties, we compute the derivatives of \(G_J\). A simple calculation shows that
\[
\partial_1 g_J(b,v) = \frac{v}{2(\theta_J - b)} \biggl( \frac{F_{\theta_J}'''(b)}{F_{\theta_J}'(b)} + \frac{F_{\theta_J}''(b)}{F_{\theta_J}'(b) (\theta_J - b)} - \Bigl( \frac{F_{\theta_J}''(b)}{F_{\theta_J}'(b)} \Bigr)^2 \biggr)
\]
for \(b \neq \theta_J\), and applying l'Hospital's rule shows that as \(b \to \theta_J\), \(\partial_1 g_J(b,v) \to - \frac{F_{\theta_J}''''(b)}{2F_{\theta_J}'(\theta_J)}\). This is zero since \(F_{\theta_J}''''(\theta_J) = 0\), as the distribution is unimodal and symmetric. It is easy to verify by applying l'Hospital's rule twice that \(\partial_1 g_J(\theta_J, v) = 0\), so that \(\partial_1 g_J\) exists and is continuous on \(\mathbb{R} \times [0,\infty)\). As it is clear that \(\partial_2 g_J\) exists and is continuous on \(\mathbb{R} \times [0,\infty)\), it follows that \(g_J\) is \(\mathscr{C}^1\). This directly implies that \(G_J\) is \(\mathscr{C}^1\) as well. 

Now, a direct computation leads to
\[
\partial_1 G_J(b,v) = F_{\theta_J}'\Bigl( \frac{b + \theta_J g_J(b,v)}{1 + g_J(b,v)} \Bigr) \frac{1 + g_J(b,v) + (\theta_J - b) \partial_1 g_J(b,v)}{\bigl( 1 + g_J(b,v) \bigr)^2}.
\]
This is always positive by assumption \ref{thm.GJ.conditionDgJ}, and property \ref{thm.GJ.increasing} follows. We record here for future computations that the above implies that
\beq{eq.D1GJ(b;0)}
\partial_1 G_J(b,0) = F_{\theta_J}'(b).
\eeq
Finally,
\[
\partial_2 G_J(b,v) = F_{\theta_J}'\Bigl( \frac{b + \theta_J g_J(b,v)}{1 + g_J(b,v)} \Bigr) \frac{(\theta_J - b) \partial_2 g_J(b,v)}{\bigl( 1 + g_J(b,v) \bigr)^2},
\]
and property \ref{thm.GJ.secondderivative} follows by evaluating at \(v = 0\).
\end{proof}

We now prove the corollaries of the Theorem. 

\begin{corollary}
Suppose that all assumptions of Theorem~\ref{thm.GJ} hold. Then \(f\colon (0,\infty) \times \mathbb{R}^3 \to (0,\infty)\) defined by
\[
f(s,b,c,v) \defeq s G_J\Bigl( b + \frac{c}{s}, v \Bigr)
\]
satisfies the following conditions for any \(s > 0\), \(b,c \in \mathbb{R}\) and \(v \geq 0\).
\begin{enumerate}
    \item\label{cor1.regularity} \(f\) is \(\mathscr{C}^1\) on \((0,\infty) \times \mathbb{R}^3\);
    \item\label{cor1.sign} \(f(s,b,c,v) \gtreqless sG_J(b,v)\) when \(c \gtreqless 0\), and in particular \(f(s,b,0,v) = sG_J(b,v)\) and \(f(s,b,0,0) = sF_{\theta_J}(b)\);
    \item\label{cor1.limits} \(f(s,b,c,v) \to 0\) as \(b\to-\infty\) and \(f(s,b,c,v) \to s\) as \(b \to \infty\);
    \item\label{cor1.D1} \(\partial_3f(s,b,0,0) = F_{\theta_J}'(b)\);
    \item\label{cor1.D2} \(\partial_4f(s,b,0,0) = \frac{1}{2} s F_{\theta_J}''(b)\).
\end{enumerate}
\end{corollary}

\begin{proof}
Properties \ref{cor1.regularity}, \ref{cor1.sign} and \ref{cor1.limits} follow respectively from properties \ref{thm.GJ.regularity}, \ref{thm.GJ.increasing} and \ref{thm.GJ.extremeinput} of \(G_J\) given in Theorem~\ref{thm.GJ}. Then, properties \ref{cor1.D1} and \ref{cor1.D2} follow by evaluating the derivatives using \cref{eq.D1GJ(b;0)} and property \ref{thm.GJ.secondderivative} of \(G_J\).
\end{proof}

\begin{corollary}
Suppose that all assumptions of Theorem~\ref{thm.GJ} hold. Then \(f \colon (0,\infty)^2 \times \mathbb{R}^5 \to \mathbb{R}\) defined by
\[
f(x,s,b,c_1,c_2,c_3,v) \defeq (xs + c_1) G_J\Bigl( b + \frac{c_2}{x} + \frac{c_3}{s}, v \Bigr)
\]
satisfies the following conditions for any \(x,s > 0\), \(b \in \mathbb{R}\) and \(v \geq 0\).
\begin{enumerate}
    \item\label{cor2.regularity} \(f\) is \(\mathscr{C}^1\) on \((0,\infty)^2 \times \mathbb{R}^5\);
    \item\label{cor2.zerocovs} \(f(x,s,b,0,0,0,v) = xs G_J(b,v)\), in particular \(f(x,s,b,0,0,0,0) = xs F_{\theta_J}(b)\);
    \item\label{cor2.D4} \(\partial_4 f(x,s,b,0,0,0,0) = F_{\theta_J}(b)\);
    \item \(\partial_5 f(x,s,b,0,0,0,0) = s F_{\theta_J}'(b)\);
    \item \(\partial_6 f(x,s,b,0,0,0,0) = x F_{\theta_J}'(b)\);
    \item\label{cor2.D7} \(\partial_7 f(x,s,b,0,0,0,0) = \frac{1}{2} xs F_{\theta_J}''(b)\).
\end{enumerate}
\end{corollary}

\begin{proof}
Property \ref{cor2.regularity} follows from property \ref{thm.GJ.regularity} of \(G_J\) in Theorem~\ref{thm.GJ}. Then, property \ref{cor2.zerocovs} follows by direct evaluation, and properties \ref{cor2.D4} to \ref{cor2.D7} follow by evaluating the derivatives using \cref{eq.D1GJ(b;0)} and properties \ref{thm.GJ.extremevariance} and \ref{thm.GJ.secondderivative} of \(G_J\).
\end{proof}

\subsection{Verification of the assumptions for a normal distribution}

Suppose that the thresholds in population \(J\) follow a normal distribution with mean \(\theta_J\) and variance \(\sigma_{\theta_J}^2\), which has density
\[
F_{\theta_J}'(b) = \frac{1}{\sigma_{\theta_J} \sqrt{2\pi}} \exp\Bigl( - \frac{(b - \theta_J)^2}{2\sigma_{\theta_J}^2} \Bigr).
\]
This distribution is indeed unimodal and symmetric. Moreover, its cumulative distribution function \(F_{\theta_J}\) is smooth, and the density \(F_{\theta_J}'\) is positive, so assumptions \ref{thm.GJ.conditionregularity} and \ref{thm.GJ.conditionpositivity} are satisfied. 

To verify the other two assumptions, we compute \(g_J\) for this distribution.
\[
F_{\theta_J}''(b) = \frac{\theta_J - b}{\sigma_{\theta_J}^3 \sqrt{2\pi}} \exp\Bigl( - \frac{(b - \theta_J)^2}{2\sigma_{\theta_J}^2} \Bigr) = \frac{\theta_J - b}{\sigma_{\theta_J}^2} F_{\theta_J}'(b),
\qquad\text{so}\qquad
g_J(b,v) = \frac{v}{2\sigma_{\theta_J}^2}.
\]
Then \(g_J(\cdot, 1) \equiv \nicefrac{1}{2\sigma_{\theta_J}^2}\) is bounded so assumptions \ref{thm.GJ.conditionboundedness} is satisfied, and it is positive so \ref{thm.GJ.conditionDgJ} is directly satisfied since \(g_J\) does not depend on its first argument.

Therefore, the conclusions of Theorem~\ref{thm.GJ} hold if the thresholds in population \(J\) follow a normal distribution.

\subsection{Verification of the assumptions for a logistic distribution}

Suppose that the thresholds in population \(J\) follow a logistic distribution with mean \(\theta_J\) and scaling factor \(s_{\theta_J}\), which has cumulative distribution function
\[
F_{\theta_J}(b) = \sigma\Bigl( \frac{b - \theta_J}{s_{\theta_J}} \Bigr)
\qquad\text{where}\qquad
\sigma(x) \defeq \frac{1}{1 + e^{-x}}.
\]
This distribution is unimodal and symmetric, and \(F_{\theta_J}\) is smooth so assumption \ref{thm.GJ.conditionregularity} is satisfied.

The logistic function \(\sigma\) has the property that \(\sigma' = \sigma(1 - \sigma)\). This property allows to compute
\[
F_{\theta_J}' = \frac{1}{s_{\theta_J}} F_{\theta_J} (1 - F_{\theta_J})
\qquad\text{and}\qquad
F_{\theta_J}'' = \frac{1}{s_{\theta_J}} F_{\theta_J}' (1 - 2F_{\theta_J}).
\]
Since the range of \(F_{\theta_J}\) is \((0,1)\), it follows that \(F_{\theta_J}' > 0\) and assumption \ref{thm.GJ.conditionpositivity} is satisfied. Then, for \(b \neq \theta_J\),
\[
g_J(b,v) = \frac{v}{2s_{\theta_J}} \frac{1 - 2F_{\theta_J}(b)}{\theta_J - b},
\]
and \(g_J\) is extended to a continuous function on \(\mathbb{R} \times [0,\infty)\) using l'Hospital's rule. Since \(g_J(\cdot,1)\) is continuous on \(\mathbb{R}\), in order to prove that it is bounded on \(\mathbb{R}\) it suffices to prove that it is bounded outside of the compact interval \([\theta_J - 1, \theta_J + 1]\), and this is easy to see: if \(\abs{b - \theta_J} > 1\), then \(\abs{g_J(b,1)} < \nicefrac{1}{2s_{\theta_J}}\) since \(F_{\theta_J}\) takes its values in \((0,1)\). This shows that assumption \ref{thm.GJ.conditionboundedness} is satisfied.

To verify the last assumption, notice that for \(b \neq \theta_J\) and any \(v \geq 0\),
\[
(\theta_J - b) \partial_1 g_J(b,v) = \frac{v}{2s_{\theta_J}} \Bigl( \frac{1 - 2F_{\theta_J}(b)}{\theta_J - b} - 2F_{\theta_J}'(b) \Bigr) = g_J(b,v) \Bigl( 1 - \frac{2(\theta_J - b) F_{\theta_J}'(b)}{1 - 2F_{\theta_J}(b)} \Bigr).
\]
Then, with \(x \defeq \nicefrac{(b - \theta_J)}{s_{\theta_J}}\),
\[
1 - 2F_{\theta_J}(b) = - \frac{1 - e^{-x}}{1 + e^{-x}}
\qquad\text{and}\qquad
F_{\theta_J}'(b) = \frac{1}{s_{\theta_J}} \frac{e^{-x}}{(1 + e^{-x})^2},
\]
so
\[
(\theta_J - b) \partial_1 g_J(b,v) = g_J(b,v) \Bigl( 1 - \frac{2x e^{-x}}{1 - e^{-2x}} \Bigr) = g_J(b,v) \Bigl( 1 - \frac{x}{\sinh x} \Bigr).
\]
Since \(\abs{x} \leq \abs{\sinh x}\) for all \(x \in \mathbb{R}\), it follows that \((\theta_J - b) \partial_1 g_J(b,v) \geq 0\) because \(g_J\) is nonnegative. Since this holds whenever \(b \neq \theta_J\), it holds for \(b = \theta_J\) by continuity, and assumption \ref{thm.GJ.conditionDgJ} holds.

Therefore, the conclusions of Theorem~\ref{thm.GJ} hold if the thresholds in population \(J\) follow a logistic distribution.

\section{Covariances from examples}
\label{apx.covs}

In this section, we provide Figures~\ref{fig.ex1.covs}, \ref{fig.ex2.covs} and \ref{fig.ex3.covs}, that show all covariances that are not variances in examples from Sections~\ref{sec.example1}, \ref{sec.example2} and \ref{sec.example3} respectively. 

\begin{figure}
\includegraphics[width=\linewidth]{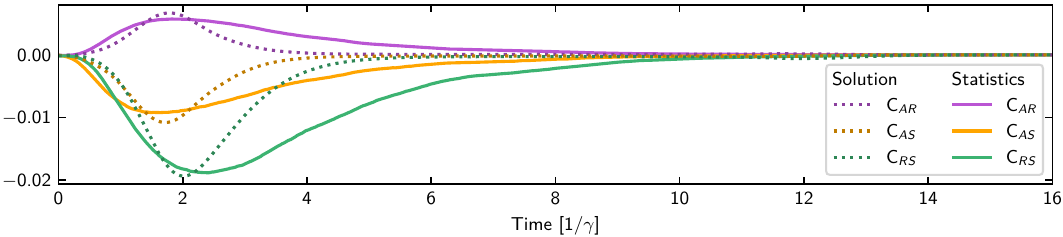}
\caption{Remaining covariances from the numerical simulations of Section~\ref{sec.example1}.}
\label{fig.ex1.covs}
\end{figure}

\begin{figure}
\includegraphics[width=\linewidth]{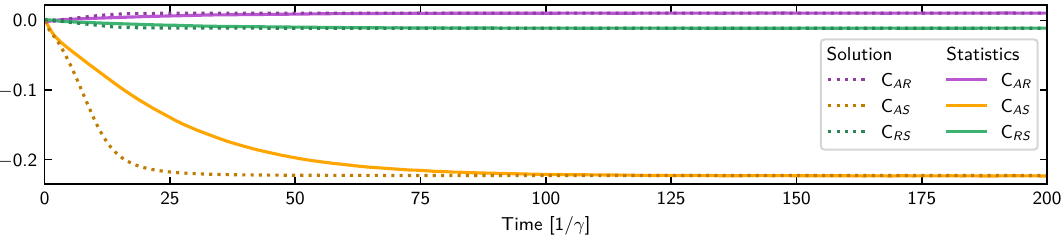}
\caption{Remaining covariances from the numerical simulations of Section~\ref{sec.example2}.}
\label{fig.ex2.covs}
\end{figure}

\begin{figure}
\includegraphics[width=\linewidth]{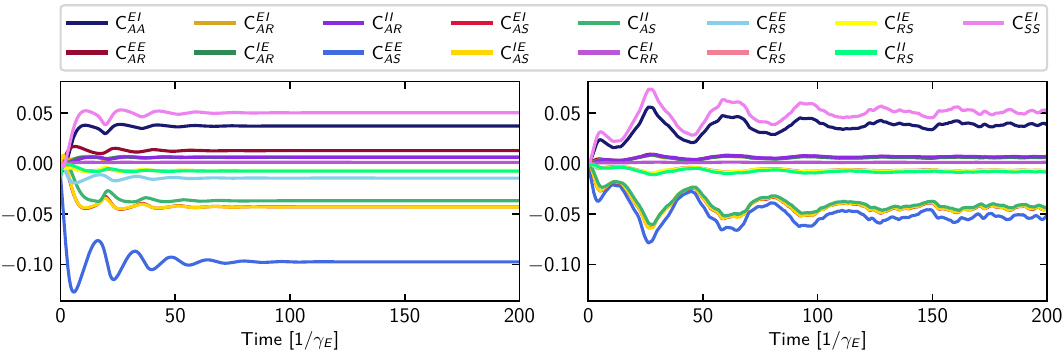}
\caption{Remaining covariances from the numerical simulations of Section~\ref{sec.example3}. On the left are components of the solution of the second-order that correspond to covariances, and on the right are statistics computed from the same simulated trajectories as in Figure~\ref{fig.ex3}.}
\label{fig.ex3.covs}
\end{figure}

\printbibliography

\end{document}